\documentclass[11pt, UKenglish]{article}
\usepackage{microtype}
\usepackage{fullpage}
\usepackage{amsmath,amsthm,amssymb}
\usepackage{xfrac}
\usepackage{mathtools}
\usepackage{enumerate}
\usepackage{bbm}
\usepackage{mathbbol}
\usepackage[sort&compress,numbers]{natbib}
\usepackage{thm-restate}
\usepackage{float} 
\usepackage[algosection,ruled,lined,linesnumbered,longend]{algorithm2e}
\usepackage{algcompatible}
\newtheorem{definition}{Definition}
\newtheorem{lemma}{Lemma}

\newtheorem{theorem}{Theorem}
\newtheorem{claim}{Claim}

\newcommand{\0}{\mathbb{0}}

\usepackage{xcolor}
\colorlet{agcolor}{blue!70!white}
\colorlet{jkcolor}{green!60!black}

\usepackage[textwidth=20mm]{todonotes}

\usepackage{comment}
\usepackage{hyperref}
\usepackage{cleveref}

\usepackage{graphicx}
\graphicspath{{./figs/}}
\usepackage{amsmath,amssymb,amsthm}

\makeatletter
\newcommand{\leqnomode}{\tagsleft@true\let\veqno\@@leqno}
\newcommand{\reqnomode}{\tagsleft@false\let\veqno\@@eqno}
\makeatother

\usepackage{wrapfig}

\usepackage{comment}
\usepackage[algosection,ruled,lined,linesnumbered,longend]{algorithm2e}
\usepackage{algcompatible}
\usepackage{color}
\usepackage{bbm}
\usepackage{mathbbol}
\usepackage{tabularx}
\usepackage{enumerate}

\usepackage{thmtools}
\usepackage{xcolor}
\colorlet{agcolor}{blue!70!white}
\colorlet{jkcolor}{green!60!black}

\usepackage[textwidth=20mm]{todonotes}
\usepackage[utf8]{inputenc}

\newcommand{\epsill}{\frac{1}{12}}   
\newcommand{\tepsill}{\frac{1}{6}}      
\newcommand{\threpsill}{\frac{1}{4}}     
\newcommand{\epsillinverse}{12}
\newtheorem{prop}[theorem]{Proposition}
\title{A Constant Factor Approximation for Directed Feedback Vertex Set in Graphs of Bounded Genus} 

\author{Hao Sun\thanks{University of Alberta, Edmonton, Canada. \texttt{hsun14@ualberta.ca}}
}

\begin{document}

\date{}
\maketitle              
\begin{abstract}
 The minimum directed feedback vertex set problem consists in finding the minimum set of vertices that should
be removed in order to make a directed graph acyclic. This is a well-known NP-hard optimization problem
with applications in various fields, such as VLSI chip design, bioinformatics and transaction processing deadlock prevention and node-weighted network design.  
We show  a constant factor approximation for the directed feedback vertex set problem in graphs of bounded genus.
\end{abstract}
\clearpage
\pagebreak
\section{Introduction}
In the  directed feedback vertex set problem (DFVS), we are given a (node-weighted) directed graph $G=(V,E)$ with costs $ c_v \ \forall v \in V$ and wish to find a minimum cost set $X$ for which $G \backslash X$ is acyclic.
DFVS is one of Karp's original 21 NP-hard problems \cite{Karp1972}. 
The DFVS problem has many applications including deadlock resolution \cite{DeadLockApp},
 VLSI chip design \cite{VLSIapp} and program verification \cite{ProgramVerifyApp}.

A 2-approximation for (undirected) FVS is given in \cite{UndirFVS2apx}. 
DFVS  has a 2-approximation in tournaments \cite{Tour2Apx}  and bipartite tournaments \cite{ZuylenBip}, is polynomial-time solvable on graphs of bounded treewidth, has a 2.4-approximation in planar graphs \cite{BY2.4} and has an  $O(\log{n}\log{ \log{n}} )$-approximation in general graphs \cite{Even1998Logn}. 
DFVS does not have an $O(1)$-approximation under the unique games conjecture \cite{DFVSisHard}. 
The genus of a graph is the minimal integer $g$ such that the graph can be drawn without crossing itself on a sphere with $g$ handles.

The following is the natural LP for DFVS and its dual, where $\mathcal{C}$ is the set of directed cycles of our graph.
 \begin{center}
\begin{minipage}{.44\textwidth}

\begin{align}
  \min ~~ & c^Tx \tag{$\mbox{P}_{  \text{DFVS} }$} \label{DFVS LP} \\
  \text{s.t.} \ \   & x(C) \geq 1 \quad \forall \ C \in \mathcal{C} \label{eq_ECT}   
  \\
          & x \geq 0  \enspace    \notag
\end{align}
\end{minipage} %
\hspace*{1ex}
\vline
\hspace*{1ex}
\begin{minipage}{.44\textwidth}
\begin{align}
  \operatorname{max} ~~ & \mathbb{1}^Ty  \tag{$\mbox{D}_{ DFVS }$} \label{DFVSlp:d2} \\
  \text{s.t.} ~~ & \sum_{C \in \mathcal{C}, v \in C} y_C \leq c_v \quad \forall v \in V(G)
           \notag \\
          & y \geq \0 \enspace . \notag
\end{align}
\end{minipage}
\end{center}

Given that constant approximations for DFVS exist for planar graphs, one naturally wonders if  DFVS admits constant approximations in bounded genus graphs.
We answer this question positively. 

\begin{restatable}{theorem}{thmDFVSapprox}\label{DFVSisPolyApx}
For any fixed genus $g$, there is a polynomial-time $O(g)$-approximation for DFVS for graphs of  genus $g$. 
Moreover, the algorithm returns a DFVS with cost $O(g)$ times the optimum solution to \eqref{DFVS LP}.
\end{restatable}
From the proof of \autoref{DFVSisPolyApx}, it is clear that the algorithm in \autoref{DFVSisPolyApx} runs in time $O(g) \text{poly}(|V(G)|)$.

For uniform costs, the dual LP \eqref{DFVSlp:d2}  is the natural LP of the \emph{ dicycle packing problem}.
The dicycle packing problem is the problem of finding the maximum number of vertex disjoint dicycles of a graph.
Schlomberg et al. \cite{schlomberg2023packing} show that the LP gap of the natural LP for dicycle packing is at most $\Omega(\frac{1}{g^2 \log g} )$ on any  graph of genus $g$.  
Our result then also implies that the minimum size of a DFVS is at most   $O(g^3 \log g)$ the size of a maximum dicycle packing. 

\subsection{Our techniques}

Informally speaking, for a (directed) graph embedded on a surface where each directed cycle bounds a region homeomorphic to the plane, one can apply the same primal-dual techniques in \cite{GoemansW1998, BY2.4} to obtain a constant factor primal-dual approximation.

In the other case, our algorithm will use the natural LP for DFVS to look for a  ``separator" \cite{BiDimDemain, FominBidim,FominBidimenKer} $S \subset V$ of cost at most a constant times the optimal DFVS such that $G\backslash S$ is of smaller genus. 
We obtain a directed cycle $C$, the removal of which results in a surface of one smaller genus. 
Traversing along the dicycle, we may define a ``left" and ``right" side of the dicycle.  
Like in \cite{Even1998Logn}, we solve the DFVS LP and use the LP values as distances.

If there is no short path leaving $C$ from the left and entering $C$ from the right and vice versa then there is a small separator $S$ such that  each dicycle of $G \backslash S$ either does not use any ``left arc" that is, an arc coming in or leaving $C$ from the left, or does not use any ``right arc" that is, an arc coming in or leaving $C$ from the right. 
$G$ with all left (resp. right) arcs deleted is of genus at least one less so inductively we can solve within a constant factor DFVS on $G$ with all left (resp. right) arcs deleted.
These two solutions together with $S$ form a DFVS of constant times more than the optimum.
If such short paths exist but all starting points of such paths and ending points of such paths are far apart the analysis is similar. 

The final case is where there are short paths $P_1, P_2$ leaving $C$ from the left and entering $C$ from the right (or vice versa) and  the starting point of $P_2$ is close to the endpoint of $P_1$.
We show that $P_1$ is ``far" from $P_2$ so to speak  (we are using directed distances so this is not the same as $P_2$ being far from $P_1$) and compute a suitable separator. 
We show that the resulting strongly connected components are of smaller genus.
We then combine approximations for different components to give an approximation for the original graph.

The presentation in this paper is focused on demonstrating linear dependence on the genus rather than optimizing the constant in \autoref{DFVSisPolyApx}.

\section{Preliminaries}


In our figures, we will use the representation of the torus by  taking the unit square $ [0,1] \times [0,1] $ and identifying the two pairs of edges $ \{0\} \times [0,1] , \ \{ 1 \} \times [0,1] $  and $  [0,1] \times \{0 \} , \   [0,1]  \times \{ 1 \} $, that is, the point $(0,p) $ is identified with $(1,p)$ and the point $(q,0)$ is identified with the point $(q,1)$.

Throughout this paper, all surfaces are orientable and  smooth and all curves are piece-wise smooth. Distances on surfaces will refer to the geodesic (shortest path on the surface) distance.  
It is well known (see for instance \cite{RydholmClassSur}) that a smooth orientable surface $Q$ is diffeomorphic to the $g$-genus torus for some $g$. For $X \subset Q$, denote $\operatorname{cl}(X)$ as the closure of $X$ in $Q$.
We henceforth assume our surface is a $g$-genus torus for some $g$.

Let us call a cycle 
of $G$, or a closed curve $C$  embedded on a surface $Q$ \emph{facial} if  it 
bounds a region $\operatorname{inside}(C)$ of the surface homeomorphic to the plane.
If the genus of $Q$ is greater than 0, call $\operatorname{inside}(C)$ the \emph{inside region} of $Q \backslash C$.
For any set $F \subset V$, define $G^F$ to be the \emph{residual graph}, that is, the subgraph of $G$ induced by those vertices that lie in a dicycle of $G \backslash F$. 





\section{Hitting the facial cycles of a digraph} 

In general, given  a (node-weighted) directed graph $G=(V,E)$ with costs $ c_v \ \forall v \in V$ a set $\mathcal{C}$ of cycles of a digraph $G$, we define the $\mathcal{C}$-hitting set problem as the problem of finding a minimum cost set $X$ such that $ X \cap C \neq \emptyset \ \forall C \in \mathcal{C}$.
In this section, we are concerned with when $\mathcal{C}$ is the set of facial cycles of our graph.

Given a digraph $G$ embedded on a surface $Q$, 
We  show how to obtain an $O(g)$-approximation for the problem of finding a  minimal hitting set for the set of facial dicycles of a graph embedded on any surface $Q$ of genus $g$.

\begin{theorem}\label{facialAlgorithm}
    For a graph $G$ embedded on a surface $Q$ of genus $g$, there is a polynomial-time $O(g)$-approximation for the problem of hitting facial dicycles of $G$. 
Moreover, the algorithm returns a DFVS with cost $O(g)$ times the optimum solution to \eqref{DFVS LP} where $\mathcal{C}$ is the set of facial dicycles.

Further, if $G$ is embedded in a way such that there exist regions $R_1,R_2$ of $Q$ homeomorphic to the open disk, such that the inside region of any facial dicycle contains at least one of $R_1,R_2$, then there is an algorithm that returns the optimum solution to \eqref{DFVS LP} where $\mathcal{C}$ is the set of facial dicycles.
\end{theorem}

%
\begin{proof}
We first show that DFVS has an $O(g)$-approximation.
If $C$ is a facial cycle bounding a face of $G$, call $C$ \emph{face minimal}.
Note that if $G$ contains a facial cycle, then it must contain a face minimal cycle by the following argument.
Let $C$ be a dicycle such that $\operatorname{inside}(C)$ is a minimal (by containment) region of $Q$.
Recall that we removed all vertices of $G$ not lying on a dicycle.
In particular, any vertex $w$ inside the region $\operatorname{inside}(C)$ must lie on a facial dicycle $ A_w $.
 $A_w$ cannot be contained entirely in $ \operatorname{cl}( \operatorname{inside}(C) ) $, as  then the region  $A_w$ would be strictly contained in $\operatorname{inside}(C)$.
 Thus, $ \operatorname{inside}( A_w)$ intersects $C$ and there is a dipath $P$ between two nodes $u,v$ of $C$.
If $C$ is not a face, then either there is a vertex $w$ inside the region $R_C$ or there is an edge $uv$ between two nodes $u,v$ of $C$ such that $g(uv) \backslash (g(v) \cup g(u))$ lies in $R_C$.
In either case, there is a dipath $P$ between two nodes $u,v$ of $C$ then $P$ together with either the $u$-$v$ or $v$-$u$ dipath in $C$ forms a cycle bounding a smaller region of $Q$, which is a contradiction.



Our algorithm proceeds as follows. 
This is a primal-dual algorithm analogous to the technique of \cite{GoemansW1998} for DFVS in planar graphs.
Given a feasible dual solution $y $ to \eqref{DFVSlp:d2}, let the \emph{residual cost} of node $v \in V$ be $ c_v - \sum_{C \in \mathcal C, v \in C} y_C$.
For $\hat{S} \subset V(G)  $, recall $G^{\hat{S}}$ denotes the subgraph  of $G$  induced by those vertices which are in a dicycle of $G \backslash \hat{S}$.  

Our primal-dual method begins with a trivial feasible dual solution $y=\0$, and the empty, infeasible hitting set  $\hat{S}= \emptyset$.

While $G^{\hat{S}}$ contains a facial cycle, increment the dual variables $y_C$ in \ref{DFVS LP} of face minimal cycles $C$ of $G$. When a node of $G$ becomes tight add it to $\hat{S}$. 
When $G^{\hat{S}}$ contains no facial cycles apply reverse deletion to $\hat{S}$ with respect to the facial cycles of $G^{\hat{S}}$, that is, we consider each node $v$ of $\hat{S}$ in the order it was added and if $ G \backslash (\hat{S} \backslash \{ v \}) $ contains no facial cycles, delete $v$ from $\hat{S}$.
Denote by $\bar{S}$ the set $\hat{S}$  at the end of the algorithm.
In other words, we apply the primal-dual method to solve the problem of hitting all facial dicycles of $G$.

\begin{algorithm}[t] 
\caption{MinWeightDirectedFVS $(G,c)$}
\label{MWFVS}
\SetKwBlock{ReverseDeletion}{Reverse-Deletion:}
\SetAlgoVlined
\SetKwInOut{Input}{Input}\SetKwInOut{Output}{Output}
\SetKw{KwDownTo}{downto}
\Input{A digraph $G=(V,E)$ with non-negative node-costs $c_v$, for each   $v \in V$. } 
\Output{A Directed FVS $S$ of $G$.}
$S= \emptyset $  \\
 \While{$G^S$ contains a facial cycle}{
Increment all dual variables $y_C$ for face minimal cycles of $G^S$.
 Add all nodes that became tight to $S$.}  
 \ReverseDeletion{
 Let $s_1,s_2,..,s_l  $ be nodes of $S$ in the order they were added. \\
   \For{ $t = l$ \KwDownTo $1$}{ \If{ $ G^{  S \backslash \{ s_t \} } $ contains no facial cycle }{$S \leftarrow S \backslash \{ s_t \} $} }  }
   \Return{$S $}
 \end{algorithm} 
Clearly, $\bar{S} $ is a feasible hitting set for the set of facial dicycles of $G$, we claim it has cost $O(g)OPT_{LP}$.
To do so we apply that standard analysis of primal-dual methods in \cite{PrimalDualMethod, GoemansW1998}.

\begin{theorem}\label{PrimalDualBasic}  (\cite{PrimalDualMethod}): Suppose $S \subset V(G)$ and $y$ is a solution to \eqref{DFVSlp:d2} output by our primal-dual algorithm such that the following holds.
\begin{enumerate}
     \item $y$ is obtained starting with the initial feasible solution $y:=0$ and incrementing some set of dual variables $\{y_C:  v \in \mathcal{C}_t \}$ uniformly and maintaining feasibility of $y$ for iterations $t=1,2,..,l$ for some $l \in \mathbb{N}$. 
     \item  For each iteration $t \in \{1,2,3,..,l \}$, the set $\{y_C:  C \in \mathcal{C}_t \}$ of incremented dual variables satisfies   $ \sum_{C \in \mathcal{C}_t } | S \cap C | \leq \beta |\mathcal{C}_t| $. 
    \item  $ \forall v \in S, \  \ \sum_{C \in \mathcal{C} \  v \in C} y_C = c_v$.
\end{enumerate} 
Then $S$ has cost at most $\beta  \sum_{C \in \mathcal{C}} y_C$, that is at most $\beta$ times the LP value. 
\end{theorem} 

Using \autoref{PrimalDualBasic}, it suffices to prove that during any iteration $t$, the face minimal cycles $\mathcal{C}_t$ of $G^{S_t}$, where $S_t$ is our current hitting set satisfies 
\begin{equation}
    \sum_{C \in \mathcal{C}_t} |\bar{S} \cap C| \leq O(g)|\mathcal{C}_t |.
\end{equation}

 Again we remove nodes of $G$ that do not lie on any dicycle.
 Denote $\bar{S}_t$ to be the nodes of $\bar{S}$ that intersect a cycle of $\mathcal{C}_t$.
 So it suffices to show $ \sum_{C \in \mathcal{C}_t} |\bar{S}_t \cap C| \leq O(g)|\mathcal{C}_t |.$

The following definition of crossing cycles was elementary to the approach by Goemans and Williamson \cite{GoemansW1998}.
\begin{definition} \label{cross}
  Fix an embedding of a planar graph.
  Two cycles $C_1,C_2$ \emph{cross} if $C_i$ contains an edge intersecting the interior of the region bounded by $C_{3-i}$, for $i=1,2$.
  That is, the plane curve corresponding to the embedding of the edge in the plane intersects the interior of the region of the plane bounded by~$C_{3-i}$.
  A set of cycles $\mathcal{C}$ is \emph{laminar} if no two elements of $\mathcal{C}$ cross. 
\end{definition}

Denote $\mathcal{C}'$ the set of facial cycles of $\mathcal{C}$.
For a node $v \in \bar{S}$, call a cycle $C \in \mathcal{C}'$ with $C \cap \bar{S} =\{v \}$ a \emph{witness} for $v$.
Since we applied reverse deletion to $\bar{S}$ at the end of the algorithm, each node of $\bar{S}$ has a witness in $\mathcal{C}'$ which is a cycle of $G^S$.

The following result about the structure of witness cycles was vital to the 3 and 2.4 approximations for DFVS in planar graphs by \cite{GoemansW1998} and \cite{BY2.4}. 
We observe that the proof in \cite{GoemansW1998} which involves iteratively applying an ``uncrossing" procedure to two witness cycles that cross yields the same result for facial cycles of graphs on surfaces.


\begin{lemma}\label{laminarwitness}\cite{GoemansW1998}
There exists a laminar family $\mathcal{A} \subset \mathcal{C}'$ of witness cycles in $G^{S_{\bar{t}} }$ for $ \bar{S}_t $.
\end{lemma}

The laminar family $\mathcal{A}$ can be represented by  a forest where $A_1$ is an ancestor of $A_2$ if the inside region of 
$A_1$ contains the inside region of $A_2$.  Add a root node $r$ to this forest, make it the parent of every maximal node of the forest and call the resulting tree $T$. 

We assign each cycle $C$ of $\mathcal{C}_t$ to the smallest node of $T$ containing $C$. Call the set of cycles assigned to $w \in T$, $\mathcal{C}_w$. We assign the nodes that $w$ and the children of $w$ are witnesses of to $w$ and call this set $\bar{S_w}$.

To bound $\sum_{C \in \mathcal{C}_t} |\bar{S}_t \cap C|$, we define the following bipartite graph.

\begin{definition}\cite{GoemansW1998}
\label{debit}
    The \emph{debit graph} for $\mathcal{C}_t$ and $S$ is the bipartite graph $\mathcal{D}_G = (\mathcal{R} \cup S, E)$ with edges $E_{\mathcal{C}_t} = \{ (C,s) \in \mathcal{C}_t \times S \mid s \in  C  \}$.
\end{definition} 
Since each $ C \in \mathcal{C}_t $ is incident to the vertices of $\bar{S}_t$ on $C$, $|\bar{S}_t \cap C|$ is the degree of $C$ in $\mathcal{D}_G$.
Summing this equality over each $ C \in \mathcal{C}_t $ yields $ \sum_{C \in \mathcal{C}_t} |\bar{S} \cap C| = E(\mathcal{D}_G)$.
By placing the node of the debit graph corresponding to $C$ inside the inside region of $C$ we can see that the debit graph is also embedded on $Q$. 
\begin{prop}\label{EulerGenusg} [Corollary of Euler's formula for graphs of genus $g$]
A (simple) bipartite graph $\bar{G}$ with at least three vertices embedded on a surface of genus $g$ satisfies 
\[ E(\bar{G}) \leq (2+g)|V(\bar{G})|-4  \]
if $G$ has two vertices then 
\[ E(\bar{G}) \leq (2+g)|V(\bar{G})|-3  \]
\end{prop}
\begin{proof}
    Euler's formula (for instance see \cite{KinseyTopology}) for graphs embedded on a surface of genus $g$ yields  
    $2-2g  =|V(\bar{G})| -| E(\bar{G})| + |F(\bar{G})|$.  Following the same method as the proof of Euler's formula for bipartite planar graphs with at least 3 vertices, (for instance see Corollary 4.2.10 of \cite{Diestel}) we observe that each face of $\bar{G}$ having at least 4 edges means $ |F(\bar{G}) |  \leq \frac{1}{2} |  |  $. 
    Thus, for $ |V(\bar{G})| \geq 3$,
$|E(\bar{G})| \leq 2|V(\bar{G})| - 4+4g \leq (2+g)|V(\bar{G})|-4  $.
If $|V(\bar{G})| \leq 2$ then $|E(\bar{G})|  \leq 1 \leq (2+g)|V(\bar{G})|-3  $.

\end{proof}
For a  node $w$ of $T$ that is not a leaf or the root,  the subgraph of $\mathcal{D}_G$ induced by $ \mathcal{C}_w \cup \bar{S}_w $ is embedded on $Q$ and further $|\mathcal{C}_w \cup \bar{S}_w|  \geq 3 $, thus by \autoref{EulerGenusg},
\begin{equation}
    | E(\mathcal{D}_G (\mathcal{C}_w \cup \bar{S}_w )) | \leq   (2+g)|\mathcal{C}_w | + (2+g) | \bar{S}_w| -4  = (2+g)|\mathcal{C}_w | +(2+g) ( \operatorname{deg}_T(w) -1) -4.
\end{equation}
For a leaf $v$ of $T$  

\begin{equation}
     | E(\mathcal{D}_G (\mathcal{C}_v \cup \bar{S}_v )) | \leq  (2+g)|\mathcal{C}_v | + 2 |  \bar{S}_v| -3 = (2+g)|\mathcal{C}_v | +2 ( \operatorname{deg}_T(v) -1) -3.
\end{equation}

For the root $r$ of $T$

\begin{equation}
    |E(\mathcal{D}_G (\mathcal{C}_r \cup \bar{S}_r ))| \leq  (2+g)|\mathcal{C}_r | + 2 |  \bar{S}_r|  = (2+g)|\mathcal{C}_r | +2 ( \operatorname{deg}_T(r) -1) .
\end{equation}

Summing these up we get 

$
\begin{array}{ll}
 |E(\mathcal{D}_G)|     & = \sum_{v \in T}  | E(\mathcal{D}_G (\mathcal{C}_v \cup \bar{S}_v ) )|  \\
     &  \leq  (2+g) | \mathcal{C} | + \sum_{v \in T} (2+g) \operatorname{deg}_T(v)   - 4|T| +l +4 \\
     & \leq  (2+g) | \mathcal{C} | + 2 ((2+g)|T | -2 )  - 4|T| +l +4  \\
     & \leq (2+g) | \mathcal{C} | +2g|T| +l  \\
     & \leq  (3+3g) | \mathcal{C} |  
     
 \end{array}$
 
 where $l$ is the number of (non-root) leaves of $T$. 
 Thus,  $\sum_{C \in \mathcal{C}_t} |\bar{S} \cap C| =   |E(\mathcal{D}_G)| \leq  (3+3g) | \mathcal{C} | $.

 This shows that $\bar{S}$ has cost $O(g)OPT_{LP}$ and hence our algorithm returns a solution of cost $O(g)OPT_{LP}$.

 Now let us show that in the case $G$ is embedded in a way such that there exist regions $R_1,R_2$ of $Q$ homeomorphic to the open disk, such that the inside region of any facial dicycle contains at least one of $R_1,R_2$, then \autoref{MWFVS} is an $8$- approximation.

 The proof works exactly the same as the general case. The  key here is to note that the inside regions of face minimal dicycles do not intersect.
 Thus, $R_1$  lies in the inside region of at most one cycle in $\mathcal{C}_t$. 
 Likewise, $R_2$  lies in the inside region of at most one cycle in $\mathcal{C}_t$.
 Since inside region of any facial dicycle contains at least one of $R_1,R_2$, $ | \mathcal{C}_t | \leq 2  $.
 Again \autoref{laminarwitness} holds. 
For a facial dicycle $A$, denote $\operatorname{inside}_{\mathcal{C}_t} (A) $ the set of cycles of $\mathcal{C}_t$ that lie in the closure of the inside region of $A$.

\begin{lemma}
    There do not exist distinct $A_1,A_2,A_3 \in \mathcal{A}$ such that $\operatorname{inside}_{\mathcal{C}_t} (A_1)=  \operatorname{inside}_{\mathcal{C}_t} (A_2) = \operatorname{inside}_{\mathcal{C}_t} (A_3)$.
\end{lemma}
\begin{proof}
    Suppose such $A_1,A_2,A_3$ existed. Since they are laminar we may assume w.l.o.g that $ A_1 $ is contained in the closure of the inside region of $A_2$ and $A_2$ is contained in the closure of the inside region of $A_3$.
    Let $v_i$ be the hit node that $A_i$ is the witness of.
    Note that $v_2$ does not lie on $A_1$. 
    Thus, as $v_2$ lies outside $\operatorname{inside}(A_1)$, it lies outside the closure $ \operatorname{cl}(\operatorname{inside}(A_1))$ of $\operatorname{inside}(A_1)$.
    So $v_2$ lies in $Q\backslash \operatorname{inside}(A_3) $. 
    Thus, $v_2$ does not lie on any cycle of $\operatorname{inside}_{\mathcal{C}_t} (A_3)$.
Also $v_2$ does not lie on $  A_3 $.
Thus, as $v_2$ lies inside  $ \operatorname{cl}( \operatorname{inside}(A_3) )$, it lies inside $ \operatorname{inside}(A_3) $. Thus, $v_2$ does not lie on any cycle of $  \mathcal{C}_t \backslash \operatorname{inside}_{\mathcal{C}_t} (A_3) $.

This implies that $v_2$ does not lie on any cycle of $\mathcal{C}_t$, which is a contradiction.
        
\end{proof}
 This implies that $ |\bar{S}_t| = | \mathcal{A} | \leq 2(2^{|\mathcal{C}_t|}) \leq 8$.
 Thus, $\sum_{C \in \mathcal{C}_t} |\bar{S}_t \cap C|  \leq |\bar{S}_t| |\mathcal{C}_t| \leq 8  |\mathcal{C}_t|$.
 This shows \autoref{MWFVS} is an 8-approximation.
\end{proof}






    

\section{Solving the case of no facial cycles}
We now show   the LP gap of the natural LP \eqref{DFVS LP} for $G$ has integrality gap $O(g)$ in the case $G$ contains no facial cycles. 
This will allow us to derive an $O(g)$-approximation for the general case by first using \autoref{facialAlgorithm} to obtain a hitting set $S$ for the set of facial cycles of cost at most $ O(g) OPT$ and then obtaining a hitting set $\bar{S}$ for the remaining dicycles.

\begin{lemma}\label{DFVSgapGen}
Suppose $G$ is a digraph embedded on a surface $Q$ of some fixed genus $g$ and there is no facial dicycle of $G$. 
Then the LP gap of the natural LP \eqref{DFVS LP} for $G$ 
has integrality gap $O(g)$.
\end{lemma}
\begin{proof}
We prove the statement by induction on the genus $g$.
The case $g=0$ is trivial because all cycles in planar graphs are facial.
Suppose the statement is true for $g=g'$. 
Let $Q$ be a surface of genus $g $,    
Let $G$ be a digraph embedded on $Q$.

First, while the optimal solution $\bar{x}$ to \eqref{DFVS LP} has a vertex $v$ with value $\bar{x}_v \geq \frac{1}{24}$ add v to our temporary hitting set $F$.
Formally initialize $F = \emptyset$.  
While the optimal solution $\bar{x}$ to \eqref{DFVS LP} for $G^F$ contains a value $\bar{x}_v$ which is $\frac{1}{24}$ or more add $v$ to $F$.

Let $F$ denote the final set obtained.
Let $\hat{x}$ be an optimal extreme point solution for the DFVS LP \eqref{DFVS LP} for $G^F$, so $\hat{x}_v < \frac{1}{24} \ \ \forall v \in V(G^F)$.
Standard results in iterative rounding, see for instance page 14 of \cite{lau_ravi_singh_2011}, show $F$ has cost at most 24 times the optimal value of our LP.

We now seek to define (integral) distances on $G^F$.
By standard LP theory, $\hat{x}$ has rational coordinates. 
Let $N \in \mathbb{Z}_{>0}$ be such that $N\hat{x}$ and $\epsill N$ are integral, call $N \hat{x}_v$, the weight of $v$. 
Define the weighted distance of path  $P=v_0,v_1,..,v_l$, $\omega(P)$ to be $ \omega(P):= \sum_{i=0}^{l-1} N \hat{x}_{v_i}$.
For a subgraph $H$ of $G^F$ define the weighted distance $d_{\{ \omega, H \} } (u,v) $ from $u$ to $v$ the minimum weight of a $u$-$v$ path in $H$. 
Define $ d_{ \omega} := d_{\{ \omega, G \} }$.
For $U,W \subset V(G)$, define $d_{ \omega, H } (U,W) := \min_{  u \in U,  w \in W  } \ d_\omega (u,w) $. 
Define $  d_{ \omega } (U,W) = d_{ \omega, G } (U,W)$.
 Define the weighted distance of a closed walk $P'=v_0,v_1,..,v_l v_0$, $\omega(P')$ to be $ \omega(P'):= \sum_{i=0}^{l} N \hat{x}_{v_i}$.
 The results in this paper could also be shown by instead defining the weight of each vertex to be $\hat{x}_v$ and instead defining the layers (see later) to be the vertices at distance a multiple of $1/N$ from a given set of vertices. 
 Since $\hat{x}$ is feasible the following result holds.
\begin{prop}
The weighted distance of any (directed) closed walk $P'$ is at least $N$.
\end{prop}


Since $\hat{x}$ is optimal, there exists a dicycle $ C^{1}:= v_1, v_2,.., v_{l'}$ such that $\sum_{v \in  C^{1}} \hat{x}_v =1$.
The motivation of our definition of weighted distance comes from \cite{Even1998Logn}.
In  \cite{Even1998Logn}, they also scale the LP values of \eqref{DFVS LP} so that the resulting values are integer.  
For any vertex $v$ with $\hat{x}_v=0$, they ``bypass" the vertex, that is, for each out neighbour $u$ of $v$ and in neighbour $w$ of $v$, they add the edge $wu$ to the graph and when they have done this for all neighbours, they delete $v$ from the graph.
For any vertex $v$ with $N\hat{x}_v >1` $ they replace $v$ by a ``chain" of $N\hat{x}_v >1` $   vertices  $v_1 \rightarrow v_2 \rightarrow..\rightarrow v_{N\hat{x}_v}$, that is, for $i=1,2,.., N\hat{x}_v-1$, $v_i v_{i+1}$ is an edge.  $ w v_1 $ and $v_{N\hat{x}_v} u$ are edges for each in neighbour $w$ and out neighbour $u$.
Call this graph $H$.

For any $W \subset V(H)$ they define ``layers"  $L_i = \{ v \in H: \ \ d_H(W,v)=i \} $ the nodes at distance $i$ from $W$.
They show that the cost of all layers $L_0,L_1,...$ is  $\sum_{v \in V(G)} N\hat{x}_v$.
This is very useful for us as we will use this to show that one layer in $L_1,... L_m$ has cost at most  $ \frac{1}{m} \sum_{v \in V} N\hat{x}_v$.
However, the bypassing operation and replacing a node with a chain operation of \cite{Even1998Logn} do not preserve the genus of the graph.
We instead define the notion of weighted distance $d_\omega$.  Denote the $i$-th layer from $W$ as $L_i:= \{  v \in V : i   \geq d_\omega( W, v)  > i - \omega(v)   \}$  the set of nodes for which the distance from $W$ to $v$ is at most $i$, but for which the distance plus the weight of $v$ is more than $i$. 
One can see that $v$ lies in $N \hat{x}_v$  different $L_i$, which is analogous to how $H$ defined in \cite{Even1998Logn} contains   $N \hat{x}_v$  copies of $v$ each lying in different layers as well.
In particular, a node of weight 0 does not lie in any $L_i$, which is analogous to how a vertex of weight 0 is bypassed in \cite{Even1998Logn}.

Consider the embedding of $C^{1}$ on our surface.
Given a subgraph $W$ of $G$, denote by $ g(W) $ the subset of our surface occupied by a vertex or edge of $W$.
We want to define a ``small" neighbourhood around $g(C^{1})$, not containing any vertices outside $C^{1}$, which we divide up into ``left" of $g(C^{1})$ and ``right" of $g(C^{1})$, which we do using the following propositions.
These are slightly informal statements of the exact propositions we require, the precise statements appear in \autoref{TopoAppendix}.

\begin{prop}(Informal statement of \autoref{LRsideFormal} and \autoref{LRsideFormalpt2})\label{LRside}
    Given a closed continuous non-self-intersecting curve $ C'$ embedded on an orientable surface $Q$, we may partition a small open neighbourhood about $C$ into a ``left" $L$ and ``right" $R$.  For any curve $f:[0,1] \rightarrow Q $ disjoint from $C'$ except at
$f(1)$ the partition allows us to say that $f$ "reaches" $C'$ from either the left or right. 
\end{prop}

\begin{prop}(Informal statement of corollary of \autoref{curvetype})\label{LRside2}
Let $C'$ be a non-facial closed curve.
    If a curve $h:[0,1] \rightarrow Q$  ``leaves" $C'$ at a point $h(0) \in C'$ from the left and reaches $C'$ at a point $h(1) \in C'$ from the right, then $ h([0,1])  $ together with a subcurve of $C'$ from $h(0)$ to $h(1)$ is a non-facial closed curve.
\end{prop}

 

We defer the proofs of \autoref{LRside} and \autoref{LRside2} for now.
We apply \autoref{LRside} to $g(C^{1})$. 
Let $L,R$ be as in \autoref{LRside} so that each $g(e)$ for $e \in E(G) \backslash E(C^{1})$ is disjoint from at least one of $L,R$ and for each $e \in E(G \backslash C^1)$, $g(e)$ is disjoint from both $L,R$.
\begin{figure}
    \centering
    \includegraphics[scale=0.3]{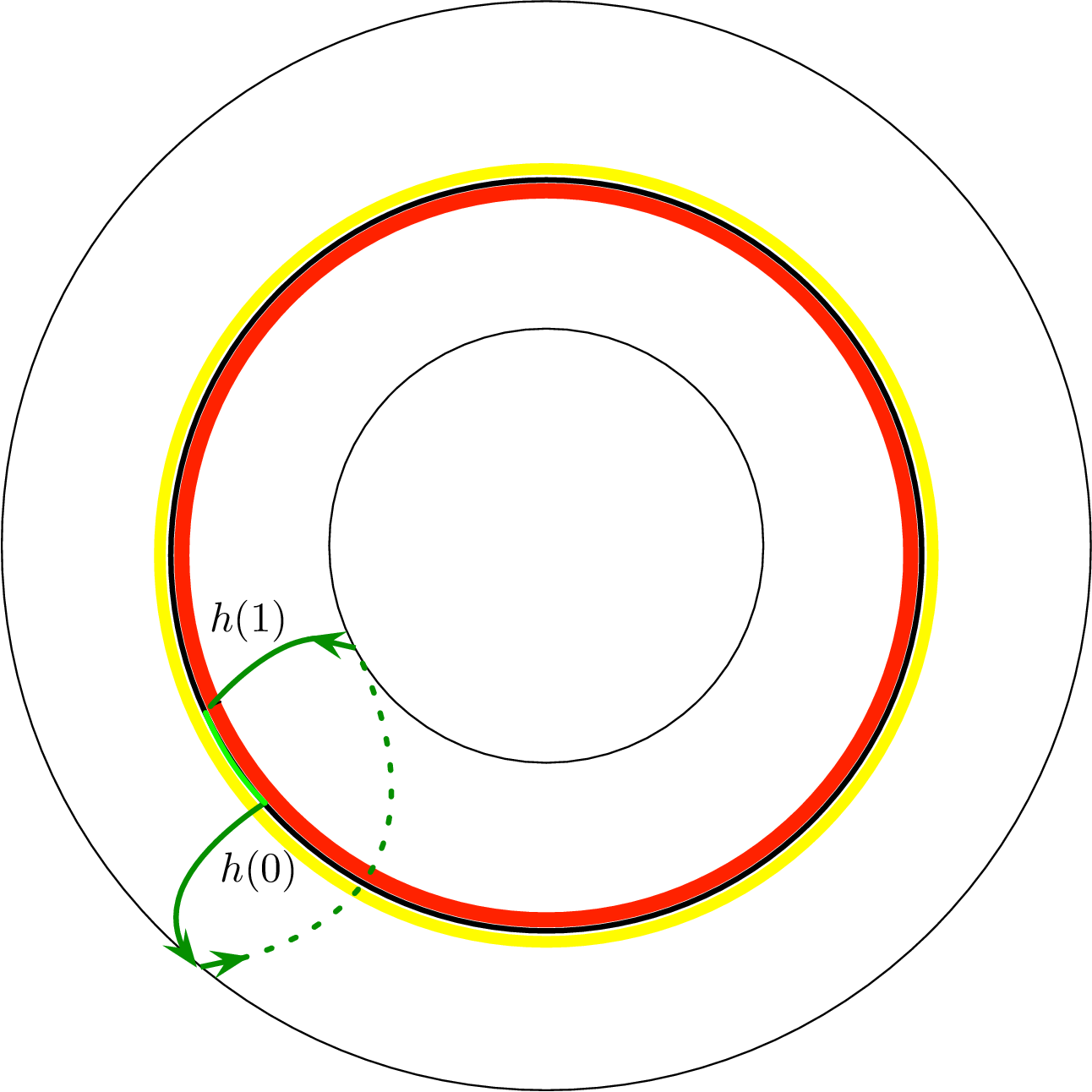}
    \caption{$L$ and $R$ from \autoref{LRside} in yellow and red respectively curve $C'$ depicted in black. The curve $h$ leaving $C'$ from the left and entering from the right is depicted in dark green. The closed curve formed by $h$ and the subcurve of $C'$ between $h(0)$ and $h(1)$ depicted in light green forms a non-facial closed curve. }
    \label{fig:my_label}
\end{figure}
For each arc $uv$ of $G^F$ with exactly one endpoint $v$  on $C^{1}$, $g(uv)$ can be parameterized by a (continuous) curve $f:[0,1] \rightarrow g(uv)$ with $f(0)=g(u)$, $f(1)=g(v)$.
If $f$ reaches $g(C^{1})$ from the left we say that $uv$ reaches $C^{1}$ from the left, otherwise, we say $uv$ reaches $C^{1}$ from the right.

Let $u'_{i,1}, u'_{i,2},.., u'_{i,l_i}$ be the out neighbours of $ v_i$  such that the edges $ u'_{i,t'} v_i$  reach $v_i$ from the left, that is, the arc obtained from reversing the arc $ v_i u'_{i,t'} $ of our graph reaches $v_i$ from the left.
Let $ w'_{i,1}, w'_{i,2},.., w'_{i,z_i}$ be the in neighbours of $v_i$ such that the edges $w'_{i,t'} v_i$ reach $v_i$ from the right.
Subdivide each edge $v_i u'_{i,t}$ into a path $v_i u_{i,t} u'_{i,t}$ and each edge $w'_{j,t'} v_j$ into a path $w'_{j,t'} w_{j,t'} v_j$ and give the new vertices $w_{j,t'} , u_{i,t}$ infinite cost. There is a natural embedding of our new graph on our surface by placing each $u_{i,t}$ where the midpoint of the curve $g(v_i u'_{i,t})$ was embedded and likewise for $w_{j,t'}$.
By abuse of notation, we continue to call our graph $G$ and define $\hat{x}_{u_{i,t}}= \hat{x}_{w_{j,t'}} =0$ for all $u_{i,t}$, $w_{j,t'}$.
Denote $U:= \cup_{i=1}^{l'} \{  u_{i,1}, u_{i,2},.., u_{i,l_i}  \}$ and  $W:= \cup_{i=1}^{l'} \{  w_{i,1}, w_{i,2},..,  w_{i,z_i}  \}$.
For $X \subset [l']$, denote $U_X:= \cup_{i \in X} \{  u_{i,1}, u_{i,2},.., u_{i,l_i}  \}$, $V_X= \{ v_i :  \ \ i \in X  \}$ and $W_X:= \cup_{i \in X} \{  w_{i,1}, w_{i,2},.., w_{i,z_i} \}$.


Let $\tau_-:= \{ i \in [l']:  \exists  w_{i,t'} \in W, \ \ \exists  u_{j,t} \in U :  d_{\omega, G^F \backslash C^{1}}( u_{j,t} , w_{i,t'} )  < \epsill N  \}$
the first indices of the set of vertices of $W$ of weighted distance at most $\epsill N$ from $U$ in $G^F \backslash C^{1}$.
Let $\tau_+:=   \{ j \in [l']: \exists u_{j,t} \in U, \ \  \exists  w_{i,t'} \in W :  d_{\omega, G^F \backslash C^{1}}( u_{j,t} , w_{i,t'} ) < \epsill N  \}  $ the first indices of the set of vertices of $U$ that can reach $W$ in $G^F \backslash C^{1}$ with a path of weighted distance at most $\epsill N$.

\begin{claim}\label{InOutFar1}
If $d_\omega( V_{\tau_-}, V_{\tau_+} ) > \epsill N$, then we can find $S \subset V$, $c(S) = O(1)OPT_{LP}$, where $OPT_{LP}:= \sum_{v \in V} c_v x_v$ is the value of the optimal fractional solution, such that  any strongly connected component of $G^F \backslash S$ does not contain a directed path from $U$ to $W$ in $G \backslash C^{1}$.

If $d_\omega( V_{\tau_-}, V_{\tau_+} ) \leq \epsill N$, then the LP gap of the natural LP \eqref{DFVS LP} for $G^F$ has integrality gap $O(g)$.
\end{claim}
\begin{proof}

    Suppose $d_{\omega, G \backslash C^{1}}( V_{\tau_-}, V_{\tau_+} ) > \epsill N$.
For $i=0,.., \epsill N$ let  $S_i:= \{  v \in V : i   \geq d_{\omega, G^F \backslash C^{1}}( U \backslash U_{\tau_+ }, v)  > i - \omega(v)   \}$ denote the set of vertices of $V$ that are at weighted distance $i$ from $ U \backslash U_{\tau_+ } $  in $G^F \backslash C^{1}$.
(see \autoref{LayersSi}).
Since $d_{\omega, G^F \backslash C^{1}}( U \backslash U_{\tau_+ } , W)  > \epsill N   $, for  $i=0,.., \epsill N$,  $ W \cap S_i = \emptyset $ and  $ W $ is not reachable from $U \backslash U_{\tau_+ }$ in $ (G^F \backslash C^{1}) \backslash S_i$ for any $i$.

\begin{figure}
    \centering
    \includegraphics[scale=0.4]{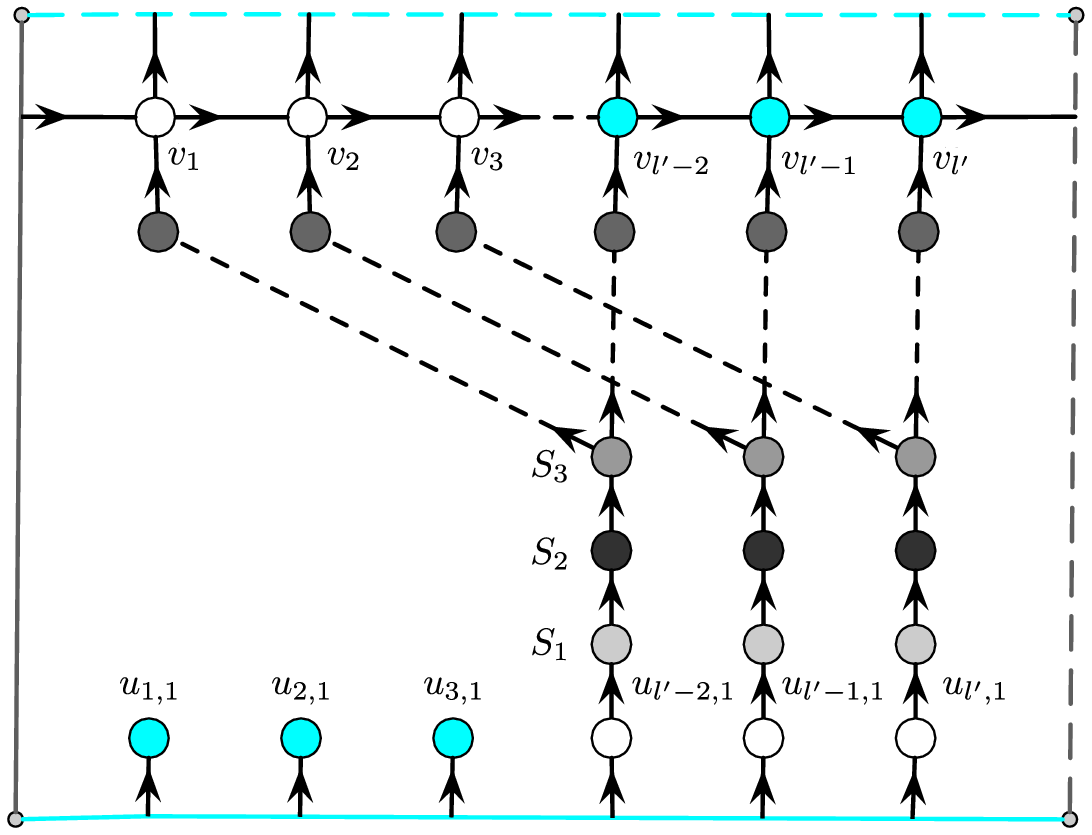}
    \caption{Nodes of $ U_{\tau_+} $ and $V_{\tau_-}$ shown in blue. }
    \label{LayersSi}
\end{figure}

Since each $v$ can lie in at most $\omega(v)$ $S_i$, $\sum_{i=0}^{\epsill N} c(S_i) \leq N \cdot OPT_{LP} $.
Let $S'$ be the  $S_i$ of minimum cost. 
For $i=0,.., \epsill N$ let  $T_i:= \{  v \in V : i  > d_{\omega,  G^F \backslash C^{1}}( v, W \backslash W_{  \tau_-  }  ) -\omega(v), \ \  d_{\omega,  G^F \backslash C^{1}}( v, W \backslash W_{  \tau_-  }  )  \geq i  \}$.   Since for $v \in  U$, $d_{\omega, G^F \backslash C^{1}}( v,  W \backslash W_{  \tau_-  } ) - \omega(v) = d_{\omega}( v,  W \backslash W_{  \tau_-  } ) > \epsill N   $,  $ U \cap T_i = \emptyset $ for $i=0,.., \epsill N$. 
Hence $ W \backslash W_{  \tau_-  } $ is not reachable from $U $ in $ ( G^F \backslash C^{1}) \backslash T_i$ for any $i$. 
Let $T'$ be the $T_i$ of minimum cost.

Finally, let $Y_i:= \{  v \in V:  \ i  \geq d_\omega(  V_{\tau_-} ,v )  > i -\omega(v) \}$ the set of vertices of weighted distance $i$ from $V_{\tau_-}$.
By assumption $d_\omega(V_{\tau_-}, V_{\tau_+}) > \epsill N$ and hence $V_{\tau_+}$ is not reachable from $V_{\tau_-}$ in $G^F \backslash Y_i$ for any $i=1,2,.., \epsill N$.
Let $Y'$ be the $Y_i$ of minimum cost.

Let $S:= S' \cup T' \cup Y'$.   
We claim no  strongly connected component $K'$ of $G^F \backslash S$  contains a directed path from $U$ to $W$ in $G^F \backslash C^{1}$.
Suppose for a contradiction that
 some strongly connected component $K'$ of $G^F \backslash S$  contains a directed path from some $u_{i,t}  \in U$ to some $ w_{j,t'} \in W $. 

If $j \notin \tau_-$, then $w_{j,t'}$ is not reachable from $U$ in $G^F \backslash S$. If $i \notin \tau_+$, then $W$ is not reachable from $u_{i,t}$ in $G^F \backslash S$. Thus, if either $j \notin \tau_-$ or $i \notin \tau_+$ then there is no path from $u_{i,t}$ to $w_{j,t'}$ in $G^F \backslash S$.
 Thus, $j \in \tau_-$ and $i \in \tau_+$. 
 As $K'$ is strongly connected, this implies that $G^F \backslash S$ contains a path from $V_{\tau_-}$ to $V_{\tau_+}$, which is not possible.

\par Now suppose that $d_\omega( V_{\tau_-}, V_{\tau_+} ) \leq \epsill N$.
Let $i \in \tau_-$ and $j \in \tau_+$ be such that $d_\omega(v_i,v_j) \leq \epsill N$.
Let $ P_1, P_2, P_3$ be  $ u_{a,t} $-$v_i$,  $u_{j,t'} $-$ v_b$ and $v_i$-$ v_j$ paths of weight at most $\epsill N$, with the second last vertices of $P_1,P_2$ being in $W$, for some $a,b$.
Such paths exist as $i \in \tau_-$ and $j \in \tau_+$. 
If $a=i$, then $P_1 v_i u_{a,t}$ is a cycle for which $\sum_{v \in P_1 v_i u_{a,t}  } \hat{x}_v <1 $ which is a contradiction. So $a \neq i$, likewise $b \neq j$.

For $i',j' \in \{1,2,..,l' \}$, let $C^1_{ (v_{i'},v_{j'}) } := v_{i'}, v_{i'+1}, v_{i'+2} ,.., v_{j'-1} v_{j'}$ (where $  v_{t} = v_{t \pmod {l'}} $) denote the directed path in $C^1$ from $v_i$ to $v_j$.
Note that  $d_\omega(v_{i'}, v_{j'}) = \omega(C^1_{ (v_{i'},v_{j'}) })  $, for otherwise there is a $v_{i'}$-$ v_{j'} $ path $P'$ of weight less than $d_\omega(v_{i'}, v_{j'})$.
Then $ C^1_{ (v_{j'} , v_{i'} ) } \cup P' $ is a directed closed walk of weight $ \omega ( C^1) - \omega(C^1_{ (v_{i'},v_{j'}) }) +  d_\omega(v_{i'}, v_{j'}) < \omega(C^{1})$.
Noting that the weighted distance of a cycle is equal to $ N \sum_{v \in C^{1}} \hat{x}_v $, we obtain
$ N \sum_{v \in C^1_{ (v_{j'} , v_{i'} ) } \cup P'} \hat{x}_v  < \omega(C^{1})  =N  $, from which it follows the sum of the $\hat{x}_v$ values along the closed walk $ C^1_{ (v_{j'} , v_{i'} ) } \cup P' $,  $\sum_{v \in C^1_{ (v_{j'} , v_{i'} ) } \cup P'} \hat{x}_v $ is strictly less than 1, which contradicts the feasibility of $\hat{x}$.

We claim $ \omega( C^1_{(v_{a},v_i)} ) , \omega( C^1_{( v_j, v_b) } ) \leq \epsill N$. 
Suppose for a contradiction that $\omega( C^1_{(v_{a},v_i)} )  > \epsill N $.
Since $ \omega(C^1)=N$, this implies that $\omega( C^1_{(v_{i+1}, v_{a-1} )} ) < N - \epsill N $. Then the cycle $ P_1  C^1_{(v_i, v_a) } v_a u_{a,t} $ satisfies $ \sum_{v \in V(P_1  C^1_{(v_i, v_a) } v_a u_{a,t})} \hat{x}_v < 1 - \epsill + \epsill =1 $ which is a contradiction.   Likewise, $ \omega( v_j, v_b ) \leq \epsill N$.

Let us show $C^1_{ (v_{a} , v_{i} ) } \cap C^1_{ (v_{j} , v_{b} ) }  = \{ v_i \} \cap \{ v_j \} $, that is the paths   $C^1_{ (v_{a} , v_{i} ) } $ and $C^1_{ (v_{j} , v_{b} ) } $ are disjoint except in the case $i=j$ when their intersection is $v_i$.
First, let us address the case $i \neq j$. Suppose for a contradiction that $C^1_{ (v_{a} , v_{i} ) } \cap C^1_{ (v_{j} , v_{b} ) }  \neq \emptyset$. Let $ v \in C^1_{ (v_{a} , v_{i} ) } \cap C^1_{ (v_{j} , v_{b} ) } $. 
Note $v \neq v_{i}, v_{j}   $ for otherwise $ C^1_{ (v_{j} , v_{i} )} P_3 $ is a closed walk of weight less than $N$.
Let $Q_1$ be a path from $v$ to $v_i$ in $C^1_{ (v_{a} , v_{i} ) }$ and $Q_2$ a path from $ v_j $ to $v$ in $C^1_{ (v_{j} , v_{b} ) }$. Then $ Q_2 Q_1 P_3$ is a closed walk of weighted distance at most $ \threpsill N $ which is a contradiction.

Now suppose that $i=j$. Suppose for a contradiction that $C^1_{ (v_{a} , v_{i} ) } \cap C^1_{ (v_{j} , v_{b} ) }  \neq \{ v_i \} $. Let $ v \in ( C^1_{ (v_{a} , v_{i} ) } \cap C^1_{ (v_{j} , v_{b} ) } ) \backslash v_i $. Let $Q_1$ be a path from $v$ to $ v_i $ in $C^1_{ (v_{a} , v_{i} ) }$ and $Q_2$ a path from $ v_i$ to $v$ in $C^1_{ ( v_{i}, v_b ) }$. Then $Q_1 Q_2$ is a closed walk of weighted distance at most $ \tepsill N$, which is a contradiction.

\begin{claim}
$d_\omega(P_2 \cup  C^1_{( v_{j+1}, v_b) }  , P_1 \cup  C^1_{(v_{a},v_i)} ) \geq  \tepsill N  $. 
\end{claim}
\begin{proof}
Suppose for a contradiction that $ d_\omega(P_2 \cup  C^1_{( v_{j+1}, v_b) }  , P_1 \cup  C^1_{(v_{a},v_i)} )  <  \tepsill N $. Let  $ s \in P_2 \cup  C^1_{( v_{j+1}, v_b) } $ and $q \in P_1 \cup  C^1_{(v_{a},v_i)}$ be such that $d_\omega(s,q) <  \tepsill N$. 
Let $P'_1$ be the directed path in $P_1 \cup  C^1_{(v_{a},v_i)}$ from  $q$ to $v_i$.  $P'_2$ the directed path in $P_2 \cup  C^1_{( v_{j}, v_b) }$ from $v_j $ to $s$ and $Q$ a path of weight at most $\tepsill N$ from $s$ to $q$. Then $ \bar{C} := v_j  P'_2 Q P'_1 P_3 $ is a closed walk such that $  \sum_{v \in \bar{C}} \hat{x}_v <1$ which is a contradiction (see \autoref{LayersnShort}). 
\par Thus, $ d_\omega(P_2 \cup  C^1_{( v_i, v_b) }  , P_1 \cup  C^1_{(v_{a},v_i)} ) \geq  \tepsill N $. Since $  d_\omega(u,v_i) \leq \epsill N $ for any $ u \in  C^1_{(v_{a},v_i)}  $, $ d_\omega(P_2, C^1_{(v_{a},v_i)} )  \geq  \epsill N $
\end{proof}

For $i=0,1,.., \epsill N $, define $ R_i := \{ v \in V : i \geq d_\omega(P_2 \cup  C^1_{( v_{j+1}, v_b) },v)  > i -\omega(v) \} $. Each vertex $v \in V$ lies in at most $\omega(v)$ $R_i$.
Let $R'$ be the  $R_i$ of the smallest cost, so $c(R') \leq {\epsillinverse} OPT_{LP}$.
Since $ d_\omega(P_2 \cup  C^1_{( v_{j+1}, v_b) }  , P_1 \cup  C^1_{(v_{a},v_i)} ) \geq \epsill N $, it follows that  $ (P_1 \cup  C^1_{(v_{a},v_i)} ) \backslash R_i$ is not reachable from $ (P_2 \cup  C^1_{( v_{j+1}, v_b) }   ) \backslash R_i $ in $G \backslash R_i$ for any $i$. 
Thus, $(P_1 \cup  C^1_{(v_{a},v_i)} ) \backslash R'$ is not reachable from $  (P_2 \cup  C^1_{( v_{j+1}, v_b) }   ) \backslash R' $ in $G\backslash R'$.

Thus, any strongly connected component of $G \backslash R'$ is either contained in $G^F \backslash ( P_1 \cup  C^1_{(v_{a},v_i)} ) $ or $G^F \backslash  ( P_2 \cup  C^1_{( v_{j+1}, v_b) } ) $.
For $i=1,2,.., \epsill N$ let $K^+_i:= \{v \in V: \ \ i \geq d_\omega(v_j,v)  > i - \omega(v) \} $ be the vertices of weighted distance $i$ from $v_j$.
Let $K'^+$ denote the $K^+_i$  of minimum cost.   

\begin{figure}
    \centering
    \includegraphics[scale=0.4]{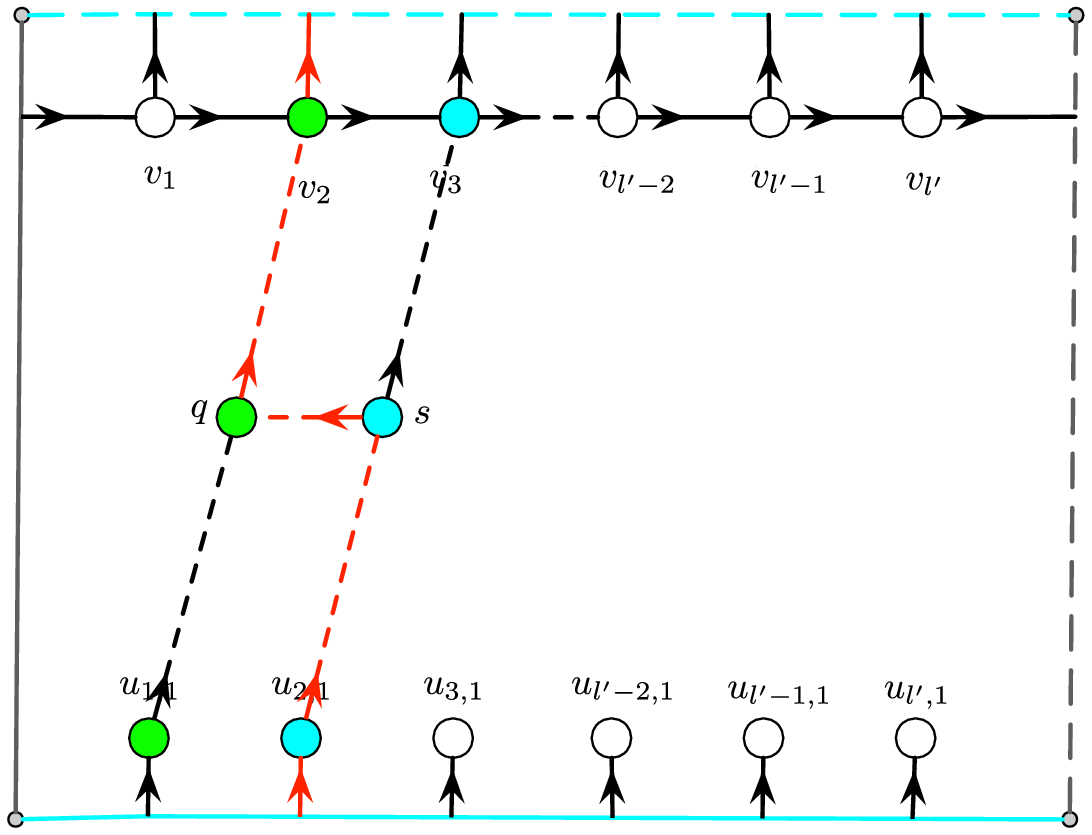}
    \includegraphics[scale=0.4]{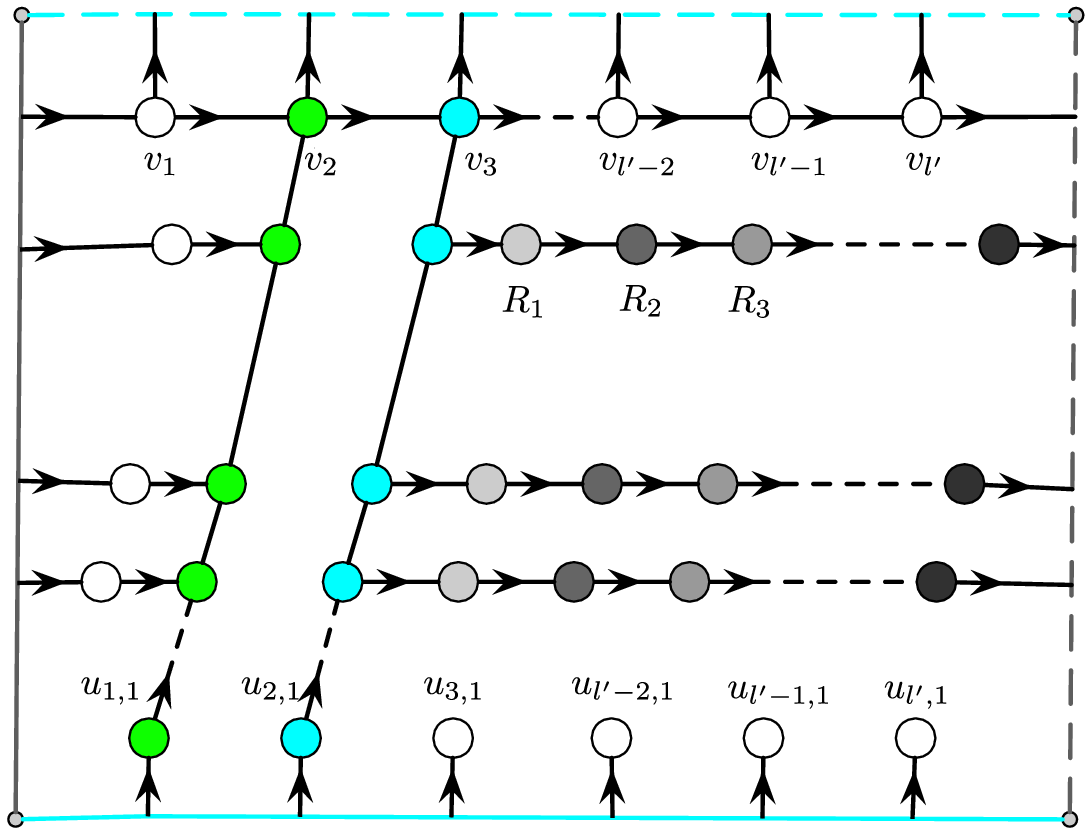} 
    \caption{On the left, there are $u_{1,1}$-$ {v_2} $ and $ u_{2,1} $-$ {v_3} $  paths (green and blue vertices respectively) of weight at most $\epsill N$  and $s$-$q$ path of length at most $\epsill N$. The red cycle would then have weight at most $N$, which is a contradiction.  On the right are the sets $R_i$, vertices at distance $i$ from $P_2 \cup  C^1_{( v_{j+1}, v_b)}$. }
    \label{LayersnShort}
\end{figure}
\begin{claim}
$v_j$ is not contained in a cycle in $G^F \backslash (R' \cup K'^+)$.
\end{claim}
\begin{proof}
Suppose that there is a cycle $ a_1,a_2,..,a_p v_j a_1 $ in $G^F \backslash (R' \cup K'^+)$. 
If $d_\omega(v_j,a_p) > \epsill N$, then $a_p$ is not reachable from $v_j$ in  $G^F \backslash (R' \cup K'^+)$.
Thus, there is a path $P_a$ of weighted distance at most $\epsill N$ from $v_j$ to $a_p$. Then the closed walk $ v_j p_a a_p v_j $ has weighted distance at most $\omega(P_a) + \omega(a_p) \leq \epsill N + \epsill N < N $ which is a contradiction.
\end{proof}

Recall that any dicycle of $G^F \backslash R' $ is contained in either $G^F \backslash ( P_1 \cup  C^1_{(v_{a},v_i)} ) $ or  $G^F \backslash ( P_2 \cup  C^1_{( v_{j+1}, v_b) } ) $. 
Since $v_j$ is not contained in any dicycle of $G^F \backslash (R' \cup K'^+)$ it follows that any dicycle of  $G^F \backslash (R' \cup K'^+)$ is either contained in $G \backslash  ( P_1 \cup  C^1_{(v_{a},v_i)}  )$ or in $G \backslash ( P_2 \cup  C^1_{( v_{j}, v_b) } ) $.
By \autoref{LRside2},  $g( P_1 \cup  C^1_{(v_{a},v_i)} )$ and $g( P_2 \cup  C^1_{( v_{j}, v_b) }  )$ are nonfacial.


\begin{definition}\label{surgery}
\cite{ArmstBasicTopo,HatcherTopo}  
Given a simple closed curve $f$ on a surface without boundary $Q$, not dividing  the surface into 2 regions,  we say $Q'$ is obtained by doing \emph{surgery} along $f$ if $Q'$ is obtained as follows.
``Thicken" $f$ to obtain a cylinder and remove this cylinder from $Q$, call this resulting surface $Q''$. The boundary of $Q''$ consists of 2 circles we ``glue" two cones $ N_1,N_2$ along these circles and call this final surface $Q'$.
\end{definition}
\begin{theorem}\label{surgreduc} (\cite{ArmstBasicTopo} p.162)
For a surface without boundary $Q$ of genus $g'$, $Q'$ obtained by \autoref{surgery} is a surface without boundary of genus at most $g'-1$. 
\end{theorem}

We apply the surgery of \autoref{surgery} to $g(P_1 \cup  C^1_{(v_{a},v_i)})$ to obtain a surface $Q'$ of genus one less than $Q$.
Let $N'_1,N'_2$ denote the two cones glued to $Q'$.
We also apply the surgery of \autoref{surgery} to $g( P_2 \cup  C^1_{( v_{j}, v_b) }  )$ to obtain a surface $\hat{Q}$ of genus one less than $Q$.
Let $\hat{N}_1, \hat{N}_2$ denote the two cones glued to $\hat{Q}$.
\begin{lemma}\label{containcone}
    Let $G$ be a graph embedded on a surface $Q$ with no dicycles. Let $h$ be a non-facial curve of $Q \backslash G$. Let $Q'$ be the surface obtained by applying the surgery of \autoref{surgery} to with respect to the curve $h$ and surface $Q$. 
    There is a natural embedding of $G$ on $Q'$ (by leaving each node of $G$ where it was in $Q$).
    Let $N_1,N_2$ denote the two cones glued to $Q'$ during the surgery process.
    Then each facial cycle of $G$ with respect to its embedding in $Q'$ contains either $N_1$ or $N_2$ in its inside region.
\end{lemma}
\begin{proof}
    Let $C$ be a facial cycle of $G^F \backslash ( P_1 \cup  C^1_{(v_{a},v_i)} ) $ with respect to its embedding in $Q'$.
 If neither of the cones $ N_1,N_2 $ are contained in the inside region of $ C  $, then $C$ is a facial cycle of $G$ with respect to its embedding in $Q$, which is a contradiction.
\end{proof}
Thus, any facial cycle of $G^F \backslash ( P_1 \cup  C^1_{(v_{a},v_i)} ) $ contains either $ N'_1 $ or $N'_2 $ in its inside region.

Now let $G_1,G_2,.., G_l$ be the strongly connected components of $G^F \backslash (R' \cup K'^+)$. 
Since any closed walk of of  $G^F \backslash (R' \cup K'^+)$ is either contained in $G \backslash  ( P_1 \cup  C^1_{(v_{a},v_i)}  )$ or in $G \backslash ( P_2 \cup  C^1_{( v_{j}, v_b) } ) $ each strongly connected component is   either contained in $G \backslash  ( P_1 \cup  C^1_{(v_{a},v_i)}  )$ or in $G \backslash ( P_2 \cup  C^1_{( v_{j}, v_b) } ) $.
If $G_i$ is contained in $G \backslash  ( P_1 \cup  C^1_{(v_{a},v_i)}  )$, then there is a natural embedding of $G_i$ in $Q'$ (obtained by leaving all nodes and edges where they are in the surgery for \autoref{surgery}).
Likewise, if  $G_i$ is contained in $G \backslash  ( P_1 \cup  C^1_{(v_{a},v_i)}  )$, then there is a natural embedding of $G_i$ in $\hat{Q}$.
Thus, for any $G_i$  contained in $G \backslash  ( P_1 \cup  C^1_{(v_{a},v_i)}  )$ by \autoref{facialAlgorithm}, there is an 8-approximation for the problem of hitting the facial cycles of $G_i$ (with respect to the natural embedding in $Q'$).
Likewise, for any $G_i$ contained in  $G \backslash ( P_2 \cup  C^1_{( v_{j}, v_b) } ) $ there is an 8-approximation for the problem of hitting the facial cycles of $G_i$.
Let $ Z_i $ be a solution for the problem of hitting facial cycles of $G_i$ of cost at most $8OPT_{LP(G_i)}$ as guaranteed by \autoref{facialAlgorithm}.

Then each $G_i \backslash Z_i$ is embedded in a surface of smaller genus with no facial cycles.

By induction, there are solutions $ A_i$ to $G_i 
\backslash Z_i $ of cost $c_{g-1} OPT_{LP(G_i 
\backslash Z_i)} $, where $c_g$ is the integrality gap of the DFVS LP for graphs of genus $g$.

Define $\hat{x}^{G_i 
\backslash Z_i } \in \mathbb{R}^{V(G_i 
\backslash Z_i)}$ as $\hat{x}^{G_i} (v) =\hat{x}_v $, where $\hat{x}$ is as in the proof of \autoref{DFVSgapGen}.
Since graphs $G_i$ are vertex disjoint, $G_i 
\backslash Z_i $ are vertex disjoint, so $\sum_{i=1}^l OPT_{LP(G_i)} \leq \sum_{i=1}^l \sum_{v \in V(G_i)} \hat{x}^G_i (v) \leq \sum_{v \in V(G)} \hat{x} =  OPT_{LP(G)}$.
Now $F \cup R' \cup K'^+ \cup ( \cup_{i=1}^l A_i )  \cup ( \cup_{i=1}^l Z_i ) $ is a DFVS of cost $ ( O(1) +c_{g-1} )OPT_{LP(G)}  =  ( O(1) + O(g-1) )OPT_{LP(G)} =  ( O(g) )OPT_{LP(G)} $.

\end{proof}

Note that the argument in \autoref{InOutFar1} is symmetric with respect to left and right and we may swap right and left to get the following result.
Let $b'_{i,1}, b'_{i,2},.., b'_{i,l'_i}$ be the in neighbours of $ v_i$  such that each edge $b'_{i,t} v_i$ reaches $v_i$ from the left and $d'_{i,1}, d'_{i,2},.., d'_{i,t'_i}$ be the out neighbours of $ v_i$ such that the edge $d'_{i,t'} v_i$ reaches $v_i$ from the right.
Subdivide each edge $b'_{j,t'} v_j$ into a path $b'_{j,t'} b_{j,t'} v_j$  and each edge $v_i d'_{i,t}$ into a path $v_i d_{i,t} d'_{i,t}$ and give the new vertices $d_{j,t'} , b_{i,t}$ infinite cost. 
There is a natural embedding of our new graph on our surface by placing each $b_{i,t}$ where the midpoint of the curve $g(v_i b'_{j,t'})$ was embedded and likewise for $d_{j,t'}$.
By abuse of notation, we continue to call our graph $G$ and define $\hat{x}_{b_{i,t}}= \hat{x}_{d_{j,t'}} =0$ for all $b_{i,t}$, $d_{j,t'}$.


Denote $  B:= \cup_{i=1}^{l'} \{  b_{i,1}, b_{i,2},.., b_{i,l'_i}  \}$, $D= \cup_{i=1}^{l'} \{  d'_{i,1}, d'_{i,2},.., d'_{i,t'_i}  \}$.
Let $\kappa_-: \{ i \in [l']:  \exists  b_{i,t'} \in D, \ \ \exists  d_{j,t} \in U :  d_{\omega, G^F \backslash C^{1}}( d_{j,t} ,  b_{i,t'} )  < \epsill N  \}$ the first indices of the set of vertices of $B$ of weighted distance at most $\epsill N$ from $D$ in $G^F \backslash C^{1}$.
Let $\kappa_+:=   \{ j \in [l']: \exists d_{j,t} \in U, \ \  \exists  b_{i,t'} \in W :  d_{\omega, G^F \backslash C^{1}}( d_{j,t} ,  b_{i,t'} ) < \epsill N  \}  $ the first indices of the set of vertices of $D$ that can reach $B$ with a path of weighted distance at most $\epsill N$ in $G^F \backslash C^{1}$.
Similarly to how we proved \autoref{InOutFar1}, we can show the following:
\begin{claim}\label{InOutFar}
If $d_\omega( V_{\kappa_-}, V_{\kappa_+} ) > \epsill N$, then we can find $T \subset V$, $c(T) = O(1)OPT_{LP}$, (recall $OPT_{LP}:= \sum_{v \in V} c_v x_v$ is the value of the optimal fractional solution), such that any strongly connected component of $G^F \backslash T$ does not contain a directed path from $D$ to $B$ in $G \backslash C^{1}$.

If $d_\omega( V_{\kappa_-}, V_{\kappa_+} ) \leq \epsill N$, then the LP gap of the natural LP \eqref{DFVS LP} for $G$ 
has integrality gap $O(1)$.
\end{claim}

We now construct a DFVS of cost at most $O(g) OPT_{LP}$. 
If either $d_\omega( V_{\kappa_-}, V_{\kappa_+} ) \leq \epsill N$ or $d_\omega( V_{\tau_-}, V_{\tau_+} ) \leq \epsill N$. Then \autoref{InOutFar1} or \autoref{InOutFar} respectively shows that that  the LP gap of the natural LP \eqref{DFVS LP} for $G$ has integrality gap $O(g)$.

Now assume both $d_\omega( V_{\tau_-}, V_{\tau_+} )  , d_\omega( V_{\kappa_-}, V_{\kappa_+} ) > \epsill N$.
Then by \autoref{InOutFar1} and \autoref{InOutFar}, there are sets $S,T$ such that any strongly connected component of $G^F \backslash (S \cup T)$ does not contain a path from $U$ to $W$ or a path from $D$ to $B$ in $G^F \backslash C^{1}$.

For any digraph $H $ define  $ \operatorname{un}(H) $ to be the underlying (undirected) graph of $H$.
Let $K$ be any strongly connected component of $G^F \backslash (S \cup T)$.  
We will prove $\operatorname{un}(K)$ does not contain any path from  $ U \cup B $ to $W \cup D$ in $ \operatorname{un}(K) \backslash C^{1} $.

\begin{prop}\label{undirPath}
If there is a (undirected) path $P=u_{i,t} q_1,q_2,.., q_t $ from some $ u_{i,t} \in U $ (resp $ u_{i,t} \in D $) in  $ \operatorname{un}(K) \backslash C^{1} $, then there is a directed path from $U$ (resp $D$) to  $q_j$ in $G^F \backslash (S \cup T \cup C^{1})  $ for any $j =1,2,.., t$.

If there is a (undirected) path $P=q_1,q_2,.., q_t w_{i,t}  $ from some $ w_{i,t} \in W $ (resp $ b_{i,t} \in B $)  in $ \operatorname{un}(K) \backslash C^{1} $, then there is a directed path from  $q_j$  to $W$  (resp $B$)  in $G^F \backslash (S \cup T \cup C^{1})  $ for any $j =1,2,.., t$.
\end{prop}
\begin{proof}
    Let  $P=u_{i,r} q_1,q_2,.., q_t $  be a path in   $ \operatorname{un}(K) \backslash C^{1} $  from some $u_{i,r} \in U$ (resp $ u_{i,r} \in D $). 
    We prove by induction $t'$ on that  there is a directed path from $U$ to $q_j$ in $G^F \backslash (S \cup T \cup C^{1})  $ for any $j =1,2,.., t'$.
    The case $ t'=1 $ is clear as each $u_{i,r} \in U$  (resp $ u_{i,r} \in D $)  only has out neighbours so the undirected edge $ \{  u_{i,r}, q_1 \} $ in  $\operatorname{un}(K)$ is directed from  $u_{i,r}$ to $q_1$.

    Now assume the statement true for $t'=t''$. For $t'=t''+1$, if the undirected edge  $ \{ q_{t'} , q_{t'+1}  \}$ is directed from  $  q_{t'}$ to $ q_{t'+1} $, then there is a directed path from $u_{i,r}$ to $q_{t'+1}$ in $G^F \backslash (S \cup T \cup C^{1})  $.

    Otherwise   $ \{ q_{t'} , q_{t'+1}  \}$ is directed from  $  q_{t'+1}$ to $ q_{t'} $.   By strong connectedness of $K$, there is a directed path $P'$ from $ q_{t'} $ to $q_{t'+1}$ in $ K \backslash (S \cup T)  $. 
    If $P'$ does not intersect $C$ then there is a directed path from $u_{i,r}$ to $q_{t'+1}$ in $G^F \backslash (S \cup T \cup C^{1})  $.
    So, assume $P'$  intersects $W$ or $B$.
    Let $P''$ denote the subpath of $P'$ from $q_{t'}$ to when $P'$ first intersects $U$ or $ B$. 
    By construction $P''$ lies in $G^F \backslash (S \cup T \cup C^{1})  $.
    As  $u_{i,r}$  lies in $U$ (resp $D$) $P''$ does not intersect $W$ (resp. $B$), as then we would have a $U$-$W$ (resp. $D$-$B$) path in $G^F \backslash (S \cup T \cup C^{1})  $.
    Thus, $P''$  is a $q_{t'}$-$B$ (resp. $q_{t'}$-$W$) path.
    Consider the subpath $Q$ of the reversal of $P'$ starting from $  q_{t'+1}$ to when  the reversal of $P'$ first intersects $D$ or $U$.
    Let $ \operatorname{rev}(Q)$ denote the reversal of $Q$. 
    Note $ \operatorname{rev}(Q)$ lies in  $G^F \backslash (S \cup T \cup C^{1})  $.
    If the starting vertex of $ \operatorname{rev}(Q)$ is in $ D$ (resp. $U$), then $ \operatorname{rev}(Q) \cup \{ q_{t'+1}  q_{t'}  \} \cup P'' $ is a $D$-$B$  (resp.  $U$-$W$) path in $G^F \backslash (S \cup T \cup C^{1})  $.
    This contradicts \autoref{InOutFar}.
    Thus, the starting vertex of of $ \operatorname{rev}(Q)$ is in $ U$ (resp. $D$).
    This implies there is a path from $U$ (resp. $D$) to   $ q_{t'+1} $ completing the induction.
    The proof of the second part is similar.    
\end{proof}

\begin{prop}\label{RightLeftPath}
    There is no (undirected) path from $ W \cup D $ to $U \cup B$ in $\operatorname{un}(K) \backslash C^{1}$.
\end{prop}
\begin{proof}
If we have a $U$-$W$ path $P=u_{i,t} q_1,q_2,.., q_t w_{j,t'} $ in $ \operatorname{un}(K) \backslash C^{1}$, then by \autoref{undirPath}, there are directed  $U$-$ q_1 $ and  $q_1$-$W$ paths $ P_1$ and  $ P_2 $  in $\operatorname{un}(K)  \backslash C^{1}$.  Then $P_1 \cup P_2$ is a directed $U$-$W$ path in $K \backslash C^{1}$ which contradicts \autoref{InOutFar}.
Thus, we do not have a $U$-$W$ path $P=u_{i,t} q_1,q_2,.., q_t w_{j,t'} $ in $K \backslash C^{1}$.  
Likewise, we do not have a $D$-$B$ path $P=u_{i,t} q_1,q_2,.., q_t w_{j,t'} $ in $K \backslash C^{1}$.  

Suppose we have a $U$-$D$ path $P=u_{i,t} q_1,q_2,.., q_t d_{j,t'} $ in $\operatorname{un}(K)  \backslash C^{1}$.
By \autoref{undirPath}, there are directed  $U$-$ q_1 $ and  $D$-$q_1$ paths $ P_1$ and  $ P_2 $  in $K \backslash C^{1}$.
Recall $U$ has no in-neighbours of in $G \backslash C^{1}$, so the edge $\{ u_{i,t}, q_1\}$ in $K$ is directed from $ u_{i,t}$ to $ q_1$.
By 2 connectedness of $K$, there is a path $P_3$ from $ q_1 $ to $u_{i,t}$. 
The only in-neighbours of $u_{i,t}$ are in $C^{1}$, thus $P_3$ intersects  $ W \cup B $.
Let $P'_3$ be the subpath of $P_3$ from $q_1$ to when it the path first intersects $ W \cup B $.
If the endpoint of $P'_3$ is in $W$, then $P_1 \cup P'_3$ is a $U$-$W$ path in $K \backslash C^{1}$.
Otherwise, if the endpoint of $P'_3$ is in $B$, then $P_2 \cup P'_3$ is a $D$-$B$ path in $K \backslash C^{1}$.
Either way this contradicts \autoref{InOutFar}.
    
\end{proof}

\begin{prop}\label{UnboundedLowerGenus} \cite{Mathoverflow}\cite{YoungsTheoremPaper}
    Suppose $G$ is a graph embedded on a surface $Q$. Let $C$ be a cycle of $G$ that does not divide $Q$ into two separate regions such that there is no edge between vertices of $C$ that is not part of $C$.  
    Define a ``left" and ``right" as in \autoref{LRside}. 
    Let $\hat{L}$, $\hat{R}$ denote the neighbours of $C$ that are ``left" or ``right" of $C$. 
    Suppose each connected component of  $G \backslash C$ only contains nodes of $\hat{L}$ or $\hat{R}$ but not both.
    There is a non-facial closed curve $h$ in   $  Q \backslash G $.    
    
\end{prop}

Applying \autoref{RightLeftPath}, we get that $G^F $  satisfies \autoref{UnboundedLowerGenus} with respect to $ C^{1}$. 
Thus, there is a non-facial closed curve $h$ in   $  Q \backslash G^F $.
We apply the surgery of \autoref{surgery} with respect to the closed curve $h$ and surface $Q$ to obtain a surface $Q'$ of lower genus.
Let $G_1,G_2,., G_l$ be the strongly connected components of $G^F \backslash  (S \cup T)$, so each $G_i$ is embeddable on $Q'$.
By \autoref{containcone} each facial dicycle of $G_i  $ contains one of the cones of $Q'$.
Hence there is an algorithm that returns a hitting set $Z_i$ to the set of facial cycles of $G_i$ of cost at most  $8OPT_{LP(G_i)}$.
By induction, there are solutions $ A_i$ to $G_i \backslash Z_i $ of cost $c_{g-1} OPT_{LP(G_i)} $, where $c_g$ is the integrality gap of the DFVS LP for graphs of genus $g$.
Define $\hat{x}^G_i \in \mathbb{R}^{V(G_i)}$ as $\hat{x}^G_i (v) =\hat{x}_v $, where $\hat{x}$ is as in the proof of \autoref{DFVSgapGen}.
Since graphs $G_i$ are vertex disjoint, $\sum_{i=1}^l OPT_{LP(G_i)} \leq \sum_{i=1}^l \sum_{v \in V(G_i)} \hat{x}^G_i (v) \leq \sum_{v \in V(G)} \hat{x} =  OPT_{LP(G)}$.
Then $ S \cup T \cup F \cup ( \cup_{i=1}^l A_i ) \cup ( \cup_{i=1}^l Z_i ) $ is a DFVS of cost $ ( O(1) +c_{g-1} )OPT_{LP(G)} =( O(1) +O(g-1) )OPT_{LP(G)} = ( O(g) )OPT_{LP(G)}$.

\end{proof}

As observed in \cite{Even1998Logn}, \eqref{DFVS LP} can be solved in polynomial-time via the ellipsoid method. 
Hence \autoref{DFVSgapGen} yields a polynomial time $O(g)$-approximation algorithm for DFVS in  graphs of genus $g$ with no facial cycle. 

\section{Statement and proofs of topological results we use}\label{TopoAppendix}
First let us prove  \autoref{UnboundedLowerGenus}.
\begin{proof}

Suppose each connected component of  $G \backslash C$ only contains nodes of $\hat{L}$ or $\hat{R}$ but not both.
Let $G_L$ and $G_R$ be the unions of the  components  of $G - C$ that only contain nodes from$\hat{L}$  and $\hat{R}$  respectively. Assume that $Q \backslash G$ contains no non-facial curve $h$.

Case 1: At least one of $G_L$ or $G_R$ is empty.

 Suppose, without loss of generality, that $G_L$ is empty.
Consider the face of $G$ that contains $C$ and intersects the left of $C$.
But, $C$ is not contractible (else it would separate the surface $Q$ into two components).
Hence, a small leftward shift of $C$  which will lie in the face $f$  will produce a non-facial curve $h$. 

Case 2: Both $G_L$ and $G_R$ are nonempty.

We claim that if a face contains vertices of $G_L, G_R$ and of $C$ 
then there is a non-facial curve in $Q \backslash G$.
Let $f$ be such a face of degree $d$.
Let $\partial f = v_0 \dotsm v_{d-1}$ be the boundary cycle of $f$, where $i \in \mathbb{Z}/d \mathbb{Z}$.
Without loss of generality,  assume that $v_{0} \in \hat{L}$ and for some $q$ $v_1,v_2,.., v_q \in C$, and $v_{q+1} \in \hat{R}$.
There is are points $p_l$ on the edge $v_0 v_1$ in the interior of $L$ and $p_R $ on the edge $v_q v_{q+r}$ in the interior of $R$.  
Let $h:[0,1] \rightarrow Q$ be a non-self-intersecting curve in $f$ from $ p_R $ to $p_L$.
Let $r_L, r_R >0$ be such that $B_Q(p_l,r_L)  \subset L, \ \  B_Q(p_R,r_R) \subset R $. 
Let $h^L, h^R$ be non-self-intersecting curves in $B_Q(p_l,r_L)$ and $B_Q(p_R,r_R) $ from $ p_l$ to $v_1$ and $ p_R $ to $v_q 
  $ respectively not intersecting $h$.
Then $h \cup h^L \cup h^R $ satisfies the conditions of \autoref{LRside2}.
Thus, $h \cup h^L \cup h^R \cup g(  p_l v_0,v_1,..., v_{q+1}  p_R ) $ does not bound a region of the closure of $f$.
As this curve lies in the closure of $f$, this implies that $f$ is not homeomorphic to an open disk.
By the classification theorem for orientable surfaces (see for instance page 87 of \cite{KinseyTopology}
), $ \operatorname{cl}(f)$ is homeomorphic to a $ m$-torus $T_m$ with a finite number of open disks removed.
Since $f$  is not homeomorphic to an open disk, $f$ contains a non-facial closed curve $h$ in its interior.

If there is a face $f$ of $G$ whose boundary contains vertices of $G_L$ and of $G_R$ (but not of $C$), then as there is no edge between $G_L$ and $G_R$, the boundary of $f$ is not connected and so $f$  is not homeomorphic to an open disk.
Just as before this implies  $f$ contains a non-facial closed curve $h$.

Now, consider the subsets $Q_L$ and $Q_R$ of $Q$ obtained by taking the union of all the vertices, edges and faces induced by $G_L \cup C$ and $G_R \cup C$, respectively.
By assumption 3, every component of $G - C$ is in $G_L$ or $G_R$, so every vertex and edge of $G$ belongs to $Q_L$ or $Q_R$.
By the subcases eliminated above under Case 2, every face of $G$ also belongs to $Q_L$ or $Q_R$ (but not both).
Then, $Q = Q_L \cup Q_R = (Q_L - C) \sqcup (Q_R - C) \sqcup C$.
This means that $C$ separates $Q$ into two components, which contradicts assumption 1.
\end{proof}

It is well known (see for instance \cite{ArmstBasicTopo} page 15) that smooth surfaces $Q$ have the property that for each $v \in Q$  there is an open ball $B_Q(v,r_0)$ of some small radius $r_0 >0$ in $Q$ and a diffeomorphism $\psi$ from  $B_Q(v,r_0)$  to the open disk $B_{\mathbb{R}^2}(0,r_0)$ of radius $r_0$ about the origin in the two-dimensional plane such that $\psi(v)=(0,0)$ and $\psi$ preserves distances from $v$, that is $  \operatorname{dist}_Q(v,x)= \| \psi(v) - \psi(x) \|$, where $\operatorname{dist}_Q(v,x)$ is the geodesic distance from $v$ to $x$ in $Q$.
For $p \in Q$, $r>0$, denote by $B(p,r)$ the open ball of radius $r$ about $p$.
We now formally state and prove what \autoref{LRside} and \autoref{LRside2} informally say. 

\begin{prop}\label{LRsideFormal}
Given a closed continuous non-self-intersecting curve $ C'$ embedded on an orientable surface $Q$.
There exist some radius $r>0$ and disjoint   subsets $L,R$ ``on each side" of $C'$   such that the set $ \{ B(v,r):  \ v \in  C' \} $ (where $B(v,r)$ is the open ball around $v$ of radius $r$ in $Q$) is contained in the union  $L \cup R \cup C'$, and for each $v \in C'$, $r' \leq r$, $L\cap B(v,r') $ and $R\cap B(v,r') $ are the two connected components of $ B(v,r') \backslash C' $. 
There is  a diffeomorphism $\phi$ from $  L \cup C' \cup R $ to a connected open neighbourhood of $C' \times \{0 \} $ in $ C' \times  \mathbb{R}$ and small $q>0$ with $  C' \times ( -q ,0) \subset \phi(L) \subset C' \times ( -\infty ,0) $, $C' \times (0,q) \subset \phi(R) \subset C' \times (0,\infty)$ and $\phi(C')= C' \times \{0 \}$.
 Further for any (piecewise smooth) curve $f:[0,1] \rightarrow  Q$ such that $f(x) \notin C'$ for any $x \in [0,1)$, $f(1) \in C'$ satisfies that for some $ \beta \in (0,1) $,  either $f( (\beta,1) ) \in L $, that is the curve ``reaches $C$ from the left" $L$ or  $f( (\beta,1) ) \in R $, that is the curve ``reaches $C'$ from the right" $R$.
 \end{prop}
 
 \begin{prop}\label{LRsideFormalpt2}      
 For a finite set of curves $f_1, f_2, f_3,.. f_{l'}, h_1,h_2,.., h_{t'} : [0,1] \rightarrow Q$ such that for each $i$, $f_i(x) \notin C'$ for any $x \in [0,1)$ and $h_i(x) \notin C'$ for any $x \in [0,1]$ we may choose $L,R,r$ above so that each curve $f_i([0,1))$ is disjoint from at least one of $L,R$ and each curve $h_i([0,1])$ is disjoint from both $L,R$. We refer to $L$ and $R$ as the left and right of $C'$ respectively.

 Further, there are curves $f_L: [0,1] \rightarrow L$, $f_R: [0,1] \rightarrow R$ which are homotopic to $C'$. Informally speaking, these are obtained by ``slightly shifting" $f$ ``left" and ``right" respectively.
\end{prop}
\begin{prop}\label{LRsideFormalpt3}  
Lastly let $h:[0,1] \rightarrow Q$ be any curve that reaches $C'$ from the right at a point $c_2=h(1)$ on $C'$, leaves $C'$ from the left at   $c_1=h(0)$, that is the curve $ \bar{\psi}(t)=h(1-t) $ reaches $C'$ at $c_1$ from the left and $h$ is otherwise disjoint from $C'$. Assume $c_1 \neq c_2$ and let $C'_{c_1,c_2}$ be a subcurve of $C'$ with endpoints $c_1$ and $c_2$.
 Then there is a curve $\hat{h} :[0,2] \rightarrow Q $ that reaches $C'$ from the right at $c_1=\hat{h}(1)$ and leaves $C'$ at a point $ c_1=\hat{h}(0) $ $\hat{h}$ is otherwise disjoint from $C'$, and a there is a homeomorphism of $Q$ that maps $\hat{h}$  to the concatenation of $h$ and $C'_{c_1,c_2}$.
\end{prop}
\begin{prop}\label{curvetype}
Let $Q$ be an orientable surface and $\phi:[0,1] \rightarrow Q$ a closed curve not dividing $Q$ into 2 regions with disjoint subsets $L,R$ ``on each side" of $\phi$ as in \autoref{LRsideFormal}. Let $c_1,c_2 \in [0,1)$, with $c_2 \geq c_1$. 
Suppose that $\phi_1: [0,1] \rightarrow Q$ is a curve with $\phi_1(0)= \phi(c_1) \ \ \phi_1(1) =  \phi(c_2) $, $\phi_1((0,1))$ is disjoint from $\phi([0,1])$ and the curve $ \phi_1([0, 0.5])$ approaches $\phi([0,1])$ from the left $L$ and   $ \phi_1([ 0.5, 1])$ approaches $\phi([0,1])$ from the right $R$. 

Then the curve $\phi_2 :[0,1] \rightarrow Q$,  $\phi_2(x) = \phi(c_2 -x)$ if    $ x \leq c_2 - c_1$ and $\phi_2(x)= \phi_1( \frac{1}{c_2 -c_1}  (x- c_2 +c_1) ) $ for $x > c_2 - c_1$, that is the curve obtained by  joining the portion of $ \phi$ from $c_1$ to $c_2$ to $\phi_1([0,1])$ does not divide $Q$ into 2 regions.
\end{prop}

\begin{proof}
We prove \autoref{LRsideFormal}, \autoref{LRsideFormalpt2}, and \autoref{LRsideFormalpt3}.
We use the following corollary of the tubular neighbourhood theorem (see for instance \cite{SilvaSymGeo}).

\begin{prop}\label{tubneighcor} [corollary of tubular neighbourhood theorem \cite{SilvaSymGeo}]
Given a curve $C'$ embedded in a surface $Q$ there is an open neighbourhood $U$ of $C$ and an open set $V$ in $C' \times \mathbb{R}$ such that there is a diffeomorphism $\phi:U \rightarrow V$ with $\phi(C')= C' \times \{0\}$.
\end{prop}
Let $w:[0,1] \rightarrow C'$ with $w(0)=w(1)$,  $w(0.5)=c_2$ be a parameterization of $C'$. 
We define distance on $C' \times \mathbb{R}$ by $\operatorname{dist}( (a_1,b_1) , (a_2,b_2) ) := ( \operatorname{dist}_Q(a_1,a_2)^2  + |b_1 - b_2 |^2  )^{\frac{1}{2}} $, where $\operatorname{dist}_Q(a_1,a_2)$ is the geodesic distance between $a_1$ and $a_2$ in $Q$.  
For each $v \in C'$ let $q_v$ be the minimum of 1 and $ \operatorname{sup} \{ q':  \  B(v,q')  \subset V \} $. $q_v$ is continuous in $v$ and $q_v >0 \ \forall v \in C'$. By compactness of $C'$, $q:= \operatorname{min}_{v \in C'}  q_v$ exists and is positive.
Now define $L' = \phi^{-1} (C' \times (-q, 0) )$, $R' = \phi^{-1}(C' \times (0, q)$. 

Define $U'= L' \cup C' \cup R'$.
Let $f: [0,1] \rightarrow Q$ be a curve with $f(x) \notin C'$ for any $x \in [0,1)$. By continuity of $f$ there is some $\beta \in (0,1)$ for which $f((\beta,1]) \subset U'$. 
We claim $ f((\beta,1)) \subset L' $ or $ f((\beta,1]) \subset R' $. 
If $ \phi( f((\beta,1)))$ contains a point in $C' \times (-\infty, 0)$ and a point in $\phi^{-1}(C' \times (0, \infty)$ then by continuity $ \phi( f((\beta,1)))$ contains a point in $C' \times \{0\} $ and hence $f((\beta,1))$
contains a point in $C'$ which is a contradiction.
Thus, either $ \phi( f((\beta,1))) \subset C' \times (-\infty, 0) $ or $ \phi( f((\beta,1))) \subset C' \times ( 0, \infty) $.
If  $ \phi( f((\beta,1))) \subset C' \times (-\infty, 0) $, then $ f((\beta,1)) \subset L' $.
If  $ \phi( f((\beta,1))) \subset C' \times ( 0, \infty) $, then $ f((\beta,1)) \subset R' $.

For each $f_i$ there exists $\beta_i$ for which $f((\beta_i,1]) \in L' $ or  $f((\beta_i,1]) \in R' $.
For each $x \in C'$ define $r'_x$ to be the supremum of all radius $r''_x$ for which the ball $B(x,r''_x)$ of radius $r''_x$   is entirely contained in $U'$ and for which $B(x, r''_x) $ is disjoint from $f_1([0,\beta_1)), f_2([0,\beta_s)), f_3([0,\beta_3)),.. f_{l'}([0,\beta_l)),$ $ h_1([0,1]),h_2([0,1]),.., h_{t'}([0,1]) $. If the supremum does not exist, set $ r'_x = \infty$.
Define $r_x = \operatorname{min} \{1, r'_x \} $. Again $r_x>0$ for all $x \in C'$ and is continuous in $x$.
Since $C'$ is a compact set, $r:= \operatorname{min}_{x \in C'} r_x $ exists and is positive. 
Note each curve $f_i([0,1))$ is disjoint from at least one of $L' \cap \{ B(x, r) : \ x \in C \}, R' \cap \{ B(x, r) : \ x \in C \}$ and each curve $h_i([0,1])$ is disjoint from both $L' \cap \{ B(x, r) : \ x \in C \}, R' \cap \{ B(x, r) : \ x \in C \}$.

Let us show that by making $r$ smaller if necessary $B(v,r) \backslash C'$ contains two connected components.
\begin{prop}
Given a (piece-wise smooth non-self-intersecting)  curve $f:[0,1] \rightarrow \mathbb{R}^2$ with   $ t_0 \in (0,1) $ there exists $ r >0$ for which $B(f(t_0),r) \backslash f([0,1])$ contains exactly two components.
\end{prop}
\begin{proof}
Let $[t_0,t_1]$ be an interval in which $f$ is smooth. 

Let $f(t) = f(t_0) + (t-t_0) \nabla f(t_0) + g(t-t_0)  $. By smoothness of $f$, $ \nabla g(t-t_0) $ is bounded for $t \in [t_0,t_1]$ and $g(t-t_0) =o(t-t_0)$. 
Differentiating $  \| f(t) -f(t_0) \|^2 = \| (t-t_0) \nabla f(t_0) + g(t-t_0) \|^2 $ we obtain
  \begin{align*}
\frac{d}{dt} \|f(t) - f(t_0) \|^2 = 2(  \nabla f(t_0) + \nabla g(t-t_0) )^t ((t-t_0) \nabla f(t_0) + g(t-t_0))  \\
= 2 ( \nabla f(t_0) + o(1) )^t ( (t-t_0) \nabla f(t_0) + o(t-t_0) )  \\
= 2 (\nabla f(t_0) + o(1))^t (t-t_0) \nabla f(t_0) +   o(t-t_0) 
  \end{align*} 
For $t$ close enough to $t_0$ the last line is positive. 
This implies that for some $t_2>t_0$,  $\| f(t) -f(t_0) \|^2$ is increasing on $[t_0,t_2]$. 
Likewise, for some $t_3<t_0$, $\| f(t) -f(t_0) \|^2$ is decreasing on $[t_3,t_0]$.

Let $r>0$ be such that $  \| f(t_0) -f(t) \| \geq 2r $ for all $t \in [0,1] \backslash \operatorname{frac}([t_3,t_2])$, where $\operatorname{frac}(x)$ is the fractional part of $x$.
Then   $  \| f(t_0) -f(t_2) \|, \| f(t_0) -f(t_3) \| \geq 2r $. 
Then since $\| f(t_0) -f(t) \| $  is increasing on $[t_0,t_2]$ there is exactly one $t_4 \in [t_0,t_2] $ with $ \| f(t_0) -f(t_4) \|  =r $. 
Likewise, there is exactly one $t_5 \in [t_3, t_0] $ with $ \| f(t_0) -f(t_5) \|  =r $.  So  $f([t_5,t_4])$ forms a simple curve in the closed ball $\bar{B}(f(t_0), r)$ with endpoints on the boundary and $f((t_5,t_4))$ lying in the interior. 
It follows from the Jordan curve theorem that   $B(f(t_0),r) \backslash f([0,1])  =B(f(t_0),r) \backslash f((t_5,t_4)) $ contains exactly two connected components.
\end{proof}
Let $v \in C'$. For small enough $r_0$, $B(v,r_0)$  is diffeomorphic to the open disk $B(0,r_0)$ in $\mathbb{R}^2$ via some diffeomorphism $\psi$ with $\psi(v)= (0,0)$.  
Let $w$ be a paramaterization of $C'$ with $w(0.5)=v$.  
Let $t_1 := \operatorname{inf} \{ t: \ w([t,0.5]) \subset B(v,r_0) \}$ and $t_2 := \operatorname{sup} \{ t: \ w([0.5,t]) \subset B(v,r_0) \}$ that is $[t_1, t_2]$ is a maximal interval for which $w([t_1,t_2]) \subset B(v,r_0) $  by continuity $ t_1 < 0.5 < t_2 $. 
Choose $r_1>0$   less than the distance from $ v $ to  $C' \backslash  w((t_1,t_2))$ and $r_1 < \operatorname{dist}_Q(v, w(t_1)), \operatorname{dist}_Q(v, w(t_2))   $.  From the previous proposition there exists $r_2>0$ such that $B_{\mathbb{R}^2}(\psi(v), r_2 ) \backslash \psi( w([t_1,t_2])  )$ contains exactly two connected components.  By making $r$ smaller than $r_1$ and $r_2$ if necessary we get that for any $0 < \hat{r} \leq r$, $B_{\mathbb{R}^2}(\psi(v), \hat{r} ) \backslash \psi ( C' )$ contains exactly two connected components.  
Thus,  $B(v, \hat{r} ) \backslash C'$ contains exactly two connected components. 

Define $L= L' \cap \{ B(x, r) : \ x \in C \}$, $R= R' \cap \{ B(x, r) : \ x \in C \}$. Each curve  $f: [0,1] \rightarrow Q$  with $f(x) \notin C'$ for any $x \in [0,1)$ satisfies $ f((\beta,1)) \subset L $ or $ f((\beta,1]) \subset R $. Each curve $f_i([0,1))$ is disjoint from at least one of $L,R$ and each curve $h_i([0,1])$ is disjoint from both $L,R$.

Since $\phi(v) = \{v \} \times \{ 0\}$, $ \phi( B(v, \hat{r}) )$ intersects both $ \{ v \} \times  (-\infty,0)  $ and $ \{ v \} \times  (0, \infty)  $.  Recall $B(v, \hat{r}) \subset L \cup R \cup C'$, so  $B(v, \hat{r})$  intersects both $L$ and $R$. Since there is no path from $L$ to $R$ in $ L \cup R \cup C' $ one component of $B(v, \hat{r}) \backslash C'$ is contained in $L$ and the other is contained in $R$.

For each $v \in C'$ let $y_v$ be the supremum of $\{ y'_v \geq 0: \ \{ v \} \times ( - y'_v, y'_v ) \subset \phi(B(v,r)) \}$. Again $y_v$ is positive and continuous in $v$ so $y:=\operatorname{min}_{v \in C'} y_v$ exists and is positive.
Then $ C' \times (-y,0) \subset \phi(L)$ and $ C' \times (0,y) \subset \phi(R) $.
Define the curves $f_L,f_R$ to be parameterizations of $\phi^{-1}(C' \times \{\frac{-y}{2}\})$  and $\phi^{-1}(C' \times \{\frac{y}{2}\})$ respectively.

Lastly given a curve $h:[0,1] \rightarrow Q$ be any curve that reaches $C'$ from the right at a point $c_2=h(1)$ on $C'$, leaves $C'$ from the left at   $c_1=h(0)$ and $C'_{c_1,c_2}$ be a subcurve of $C'$ with endpoints $c_1$ and $c_2$. Let $ j:[0,1] \rightarrow C'_{c_1,c_2}$ be a parameterization of $C'_{c_1,c_2}$ and denote $ \bar{j}:[0,2] \rightarrow Q$ by $\bar{j}(t)= h(t)$ for $t \in [0,1]$ and $\bar{j}(t)=j(t-1)$ for $t >1$. 
Informally speaking, we ``slightly shift" all points in $L \cup C' \cup R$ to the right while keeping $h([0,1]) \cap L$ to the left of $C'$.
 Let $\phi(x)=(\phi_1(x), \phi_2(x))$, $\phi^{-1}(x) = (\phi^{-1}_1(x), \phi^{-1}_2(x)  )$.



Let   $\gamma : C' \rightarrow (-y,0)$ be any continuous function such that $\gamma(c_2)=0$ and for any $t \in [0,1] $ for which $ \phi(h(t)) \in C' \times (-y,0)$,   $  ( \phi(h(t))_2 ) < \gamma( \phi(h(t))_1) < 0 $. 
Informally $\gamma$ is a curve lying to the right of $ \phi(h([0,1])) \cap  C' \times (-y,0)  $ and to the left of $[0,1] \times \{ 0 \}$.
Such $\gamma$ exists for instance define $ -\gamma(t)$ to be half of the  minimum of the distance from the point $ (t,0)  $ to  $ \phi(h([0,1])) \cap  C' \times (-y,0) $ and  $y$.

Define $ \bar{\gamma}: C' \times (-y,y) \rightarrow C' \times (-y,y) $ as follows. For $(a,b) \in C' \times (-y,y)  $ if $b < \frac{\gamma(a)}{4}$ define  $\bar{\gamma}(a,b)=(a,b)$.
If $ \frac{\gamma(a)}{4}  \leq b < 0 $, define $\bar{\gamma}(a,b)=(a, \frac{\gamma(a)}{4} + 2(b- \frac{\gamma(a)}{4})  )$.
If $ 0 \leq b \leq \frac{-\gamma(a)}{2} $, define $\bar{\gamma}(a,b)=(a, \frac{-\gamma(a)}{4} + \frac{b}{2} )$.
If $b \geq \frac{-\gamma(a)}{2}$ $\bar{\gamma}(a,b)=(a,b)$.
Informally, $\bar{\gamma}$  shifts $C' \times (-y,y)$ to the right while keeping  $\phi(h([0,1])) \cap  C' \times (-y,0) $ left of $C' \times \{ 0 \}$.


Define $\hat{\gamma}:Q \rightarrow Q$ by $\gamma $ by $\hat{\gamma}(v)=v$  if $v \notin \phi^{-1}( C' \times (-y,y) )$ and  $\hat{\gamma}(v)= \phi^{-1} \bar{\gamma}( \phi(v)  ) $ if $v \in \phi^{-1} ( C' \times (-y,y) )$.
Note $\hat{\gamma}$ is a homeomorphism from  $ \phi^{-1}( C' \times (-y,y) )$ to $ \phi^{-1}( C' \times (-y,y) )$  and   from   $Q \backslash \phi^{-1}( C' \times (-y,y) ) $ to $Q \backslash \phi^{-1}( C' \times (-y,y) ) $.  
Further, $\hat{\gamma}$ agrees on the boundary of $ \phi^{-1}( C' \times (-y,y) )$ and  $Q \backslash \phi^{-1}( C' \times (-y,y) ) $, that is for sequence $ \{a_i \}_{i=1}^\infty$  converging to a boundary point $a$ of $ \phi^{-1}( C' \times (-y,y) )$ (resp $Q \backslash \phi^{-1}( C' \times (-y,y) ) $)  the sequence  $\{ \hat{\gamma}( a_i ) \}_{i=1}^\infty$ converges to $\hat{\gamma}(a)$.
Thus, $\hat{\gamma}$ is a homeomorphism on $Q$, in fact, it turns out to be a continuous deformation.

Define $  \hat{h}: [0,2] \rightarrow Q $  by $\hat{h}(t)= \hat{\gamma}(  \bar{j}(t) )$. Then $\hat{\gamma}$ is a homeomorphism mapping $\hat{h}([0,2])$ to $\bar{j}([0,2])$.  


\end{proof}

The actual statement of the tubular neighbourhood theorem involves first defining the normal fibre $N_x $  as the quotient $ T_xQ/T_xC $  where $T_xQ$ and $T_xC$ are the tangent plane and tangent to the curve $C$ at $x$ and the normal bundle $NX$ as  $ \{ (x,v) : \ x \in C \ v \in N_x  \} $.
\begin{theorem}[tubular neighbourhood theorem]
There are open sets $U$ in $Q$ containing $C$ and $V$ in $NX$ such that there is a diffeomorphism  $\gamma:U \rightarrow V$.
\end{theorem}
Let us quickly show how the version of the tubular neighbourhood theorem in \autoref{tubneighcor} follows from the tubular neighbourhood theorem.

By orientability each point $x \in Q$ has a normal vector $n(x)$ and $n$ is continuous. 
We may parameterize $C'$ with a function $ \psi:[0,\beta] \rightarrow C' $ with derivative $\psi'(x)=1$ for some $\beta$. Define $v(x)$ to be of unit norm and positively orthogonal to  $ n(x), \phi'(x) $ that is $ n(x)^t v(x)= ( \phi'(x))^t v(x) = 0 $ and $ \begin{bmatrix}
n(x) & \psi'(x) & v(x)
\end{bmatrix}  
$  has determinant 1.  By the inverse function theorem, $v(x)$ is continuous. $n(x), \psi'(x), v(x)$ is a basis for $\mathbb{R}^3$ known as the Darboux frame.
Since $v(x)$ is orthogonal to $n(x)$ it lies in the tangent plant $T_xQ$ since $v(x)$ is orthogonal to $\phi'(x)$, $v(x), \phi'(x)$ is a basis for $T_xQ$. Thus, $ N_x = T_xQ/T_xC $ is diffeomorphic to $ \{ a v(x): \ a \in \mathbb{R}   \}$ which is diffeomorphic to $\mathbb{R}$. Thus, $NX$ is diffeomorphic to $C' \times \mathbb{R}$.

To prove \autoref{curvetype}, we first prove the following special case.

\begin{prop}\label{curvetypeCase}
Let $Q$ be an orientable surface and $\phi:[0,1] \rightarrow Q$ a closed curve not dividing $Q$ into 2 regions with disjoint subsets $L,R$ ``on each side" of $\phi$ as in \autoref{LRsideFormal}. Let $c \in [0,1)$, with $c_2 \geq c_1$. 
Suppose that $\phi_1: [0,1] \rightarrow Q$ is a curve with $\phi_1(0)= \phi(c) \ \ \phi_1(1) =  \phi(c) $, $\phi((0,1))$ is disjoint from $\phi([0,1])$ and the curve $ \phi_1([0, 0.5])$ approaches $\phi([0,1])$ from the left $L$ and   $ \phi_1([ 0.5, 1])$ approaches $\phi([0,1])$ from the right $R$. 

Then $\phi_1$ does not divide $Q$ into 2 regions.
\end{prop}
\begin{proof}
Suppose for a contradiction that $\phi_1$ divides $Q$ into 2 regions. It's a well-known result that one of the regions $Q_1$ must be homeomorphic to an open disk. Let $Q_2$ be the other region.

There exists a small radius $r_1$ for which the ball $B(\phi(c),r)$ of radius $r_1$ about $\phi(c)$ such that $B(\phi(c),r_1)$ is homeomorphic to an open disk. Let $r$ be as in \autoref{LRsideFormal} and define $r_2= \operatorname{min}\{ r,r_1 \}$.

\begin{definition}
  Given 2 curves  $f_1,f_2: [0,1] \rightarrow Q$ on a surface $Q$, with $f_1(0.5)=f_2(0.5)$ and $f_1(x) \neq f_2(y) \ \ \forall x,y \in ([0,1] \backslash \{0.5 \} )  $, that is they intersect only at $f_1(0,5)$, we say that $f_1$ \emph{crosses}  $f_2$ at $f_1(0.5)$ 
  if there exists $r_0>0$ such that for all $r \leq r_0$ 
 such that $f_2$ intersects both regions of $ B(f_1(0.5), r)  \backslash f_1([0,1])    $, where  $B(p,r)$ 
  is the open ball around $p$ of radius $r$.
  
\end{definition}

\begin{lemma}\label{topocrosslem}
For two curves $f_1,f_2$ on a surface $Q$ and $p$ a point on both curves, $f_1$ crosses $f_2$ at $p$ if and only if $f_2$ crosses $f_1$ at $p$.
\end{lemma}
\begin{proof}
Let $f_1(b_1)=p=f_2(b_2)$. Let $L_1$, $R_1$ and $r_1$ (resp $L_2$, $R_2$ and $r_2$) be the left, right and radius respectively for $f_1$ (resp. $f_2$ ) as guaranteed by \autoref{LRsideFormal}. Define $r= \operatorname{min}\{ r_1,r_2\}$. 
\begin{claim}
$f_1$ crosses $f_2$ at $p$ if and only if for all $r \geq r_0>0 $ none of $L_1 \cap B(p,r_0) $, $R_1 \cap B(p,r_0) $, $L_2 \cap B(p,r_0) $, $R_2 \cap B(p,r_0) $ is contained in another.
\end{claim}
\begin{proof}
Suppose that $f_1$ crosses $f_2$. Let  $r \geq r_0>0 $.
Let $t_L,t_R \in \mathbb{R}$ be such that $ f_2(t_L) \in L_1 \cap B(p,r_0) , \ f_2(t_R) \in R_1 \cap B(p,r_0)$. 
Since $R_1 \cap B(p,r_0), L_1 \cap B(p,r_0)$ are open, there exists $r_L, r_R$ be such that $B(f_2(t_L), r_L) \subset  L_1 \cap B(p,r_0) $ and $ B(f_2(t_R), r_R) \subset R_1 \cap B(p,r_0) $.  
Since $B(f_2(t_L), r_L) \cap L_2, B(f_2(t_L), r_L) \cap R_2, B(f_2(t_R), r_R) \cap L_2, B(f_2(t_R), r_R) \cap R_2 \neq \emptyset $, none of $L_1 \cap B(p,r_0) $, $R_1 \cap B(p,r_0) $, $L_2 \cap B(p,r_0) $, $R_2 \cap B(p,r_0) $ is contained in another.

Conversely, suppose that for any $0< r_0 \leq r$ none of $L_1 \cap B(p,r_0) $, $R_1 \cap B(p,r_0) $, $L_2 \cap B(p,r_0) $, $R_2 \cap B(p,r_0) $ is contained in another.
Then for any $0<r_0 \leq r$, if $f_2$ does not intersect $L_1 \cap B(p,r_0)$, then $L_1 \cap B(p,r_0)$ is connected in $  B(p,r_0) \backslash f_2$, that is $L_1 \cap B(p,r_0)$ is contained in one of the two components $L_2 \cap B(p,r_0) $, $R_2 \cap B(p,r_0) $ of $  B(p,r_0) \backslash f_2$. Hence $f_2$ intersects $L_1 \cap B(p,r_0)$, likewise $f_2$ intersects $R_1 \cap B(p,r_0)$.
Hence $f_2$ crosses $f_1$.


\end{proof}
From the previous claim it's clear that crossing is a symmetric relation.
\end{proof}

Define  
$\phi_2(x)= \phi( \operatorname{frac}( x-0.5+c ) ) $, where $\operatorname{frac}(x)= x - \lfloor x \rfloor$ is the fractional value of $x$ and  
$ \phi_3(x)= \phi_1( \operatorname{frac}( x-0.5 )  )$, 
that is, $\phi_2,\phi_3$ are reparameterized versions of $\phi$ and $\phi_1$. 
For some $ 0 <\beta_1 < \beta_2 < 1$ $\phi_1(x) \in L$ for $x \in (0, \beta_1)$ and   $\phi_1(x) \in R$ for $x \in ( \beta_2,1)$.

Define $\beta_1' = \operatorname{min} \{ \beta_1, 0.5 \}+0.5$ $\beta_2'=\operatorname{max} \{\beta_2, 0.5\} $.
Then $\phi_3((0.5, \beta_1' )) \subset L$ and $ \phi_3(( \beta_2',0.5 )) \subset R $. 
Since $L \cup R$ covers $\bar{B}(\phi(c),r) \backslash \phi([0,1])$, this implies that $\phi_3$ crosses $\phi_2$.

By \autoref{topocrosslem} $\phi_2$ crosses $\phi_3$.  
Let $L_{\phi_3}, R_{\phi_3}$ be the left and right of $\phi_3$ as in \autoref{LRsideFormal}. 
Since $L_{\phi_3}, R_{\phi_3}$ are connected, $L_{\phi_3}, R_{\phi_3}$ belong to different regions of $Q \backslash \phi_3$. 
Since $\phi_2$ crosses $ \phi_3 $, there exists $t_0$ for which $\phi(t_0) \in Q_1$ and $t_1$ for which $\phi(t_1) \in Q_2$. 
Let $t_2 = \operatorname{inf}  \{ a \in [0,1] : \ \exists b \in (a,1] \ \text{s.t.} \ \phi((a,b)) \subset Q_1  \} $, $t_3= \operatorname{sup} \{ b \in [t_2,1]: \ \text{s.t.} \phi((t_2,b)) \subset Q_1 \} $, that is $[t_2,t_3]$ is a maximal interval for which $\phi([t_2,t_3]) \subset Q_1$.
It follows $ \phi(t_2),\phi(t_3)$ lie on the boundary of $Q_1$, that is on $\phi_1$. This implies $t_2,t_3 \in \{ 0,1 \}$. Since $t_2< t_3$ $t_2=0$ and $t_3=1$. This implies that $\phi([0,1]))$ lies in $ Q_1 \cup \{ \phi(0) \} $  contradicting that there exists $t_1$ for which $\phi(t_1) \in Q_2$.

\end{proof}

\begin{proof} (of \autoref{curvetype})
Let $f:[0,1] \rightarrow C'$ be a parameterization of $C'$  and let $c= f^{-1}(c_1)$. Note that by \autoref{LRsideFormalpt3} the curve $\phi_2$ is homeomorphic to a curve $\phi_3$ that enters $\phi$ from the right and leaves $\phi$ from the left at the same point $f(c)$.  By \autoref{curvetypeCase}, $\phi_3$ does not divide $Q$ into 2 regions. Thus, $\phi_2$ does not divide $Q$ into 2 regions.

\end{proof}
\bibliography{main}

\begin{thebibliography}{10}

\bibitem{ArmstBasicTopo}
M.A. Armstrong.
\newblock {\em Basic Topology}.
\newblock Undergraduate Texts in Mathematics. Springer New York, 2013.
\newblock URL: \url{https://books.google.ca/books?id=NJbuBwAAQBAJ}.

\bibitem{UndirFVS2apx}
Vineet Bafna, Piotr Berman, and Toshihiro Fujito.
\newblock A 2-approximation algorithm for the undirected feedback vertex set
  problem.
\newblock {\em SIAM J. Discrete Math.}, 12:289--297, 1999.

\bibitem{BY2.4}
Piotr Berman and Grigory Yaroslavtsev.
\newblock Primal-dual approximation algorithms for node-weighted network design
  in planar graphs.
\newblock In Anupam Gupta, Klaus Jansen, Jos{\'e} Rolim, and Rocco Servedio,
  editors, {\em Approximation, Randomization, and Combinatorial Optimization.
  Algorithms and Techniques}, pages 50--60, Berlin, Heidelberg, 2012.

\bibitem{SilvaSymGeo}
A.C. da~Silva and LINK~(Online service).
\newblock {\em Lectures on Symplectic Geometry}.
\newblock Number no. 1764 in Lecture Notes in Mathematics. Springer, 2001.
\newblock URL: \url{https://books.google.ca/books?id=r9i6pXc7GEQC}.

\bibitem{BiDimDemain}
E.~D. {Demaine} and M.~{Hajiaghayi}.
\newblock The bidimensionality theory and its algorithmic applications.
\newblock {\em The Computer Journal}, 51(3):292--302, 2008.

\bibitem{Diestel}
Reinhard Diestel.
\newblock {\em Graph Theory}.
\newblock Springer Publishing Company, Incorporated, 5th edition, 2017.

\bibitem{Even1998Logn}
G.~Even, J.~(Seffi)~Naor, B.~Schieber, and M.~Sudan.
\newblock Approximating minimum feedback sets and multicuts in directed graphs.
\newblock {\em Algorithmica}, 20(2):151--174, 1998.

\bibitem{FominBidim}
Fedor~V. Fomin, Daniel Lokshtanov, Venkatesh Raman, and Saket Saurabh.
\newblock Bidimensionality and eptas.
\newblock In {\em Proceedings of the Twenty-second Annual ACM-SIAM Symposium on
  Discrete Algorithms}, SODA '11, pages 748--759, Philadelphia, PA, USA, 2011.
  Society for Industrial and Applied Mathematics.
\newblock URL: \url{http://dl.acm.org/citation.cfm?id=2133036.2133095}.

\bibitem{FominBidimenKer}
Fedor~V. Fomin, Daniel Lokshtanov, Saket Saurabh, and Dimitrios~M. Thilikos.
\newblock Bidimensionality and kernels.
\newblock {\em CoRR}, 2016.

\bibitem{DeadLockApp}
Georges Gardarin and Stefano Spaccapietra.
\newblock Integrity of data bases: {A} general lockout algorithm with deadlock
  avoidance.
\newblock In G.~M. Nijssen, editor, {\em Modelling in Data Base Management
  Systems, Proceeding of the {IFIP} Working Conference on Modelling in Data
  Base Management Systems, Freudenstadt, Germany, January 5-8, 1976}, pages
  395--412. North-Holland, 1976.

\bibitem{PrimalDualMethod}
Michel~X. Goemans and David~P. Williamson.
\newblock {\em The Primal-Dual Method for Approximation Algorithms and Its
  Application to Network Design Problems}, page 144–191.
\newblock PWS Publishing Co., USA, 1996.

\bibitem{GoemansW1998}
Michel~X. Goemans and David~P. Williamson.
\newblock Primal-dual approximation algorithms for feedback problems in planar
  graphs.
\newblock {\em Combinatorica}, 18(1):37--59, 1998.

\bibitem{DFVSisHard}
Venkatesan Guruswami and Euiwoong Lee.
\newblock Simple proof of hardness of feedback vertex set.
\newblock {\em Theory of Computing}, 12(6):1--11, 2016.
\newblock URL: \url{https://theoryofcomputing.org/articles/v012a006}, \href
  {https://doi.org/10.4086/toc.2016.v012a006}
  {\path{doi:10.4086/toc.2016.v012a006}}.

\bibitem{HatcherTopo}
Allen Hatcher.
\newblock Notes on basic 3-manifold topology.
\newblock 2007.
\newblock URL: \url{https://pi.math.cornell.edu/~hatcher/3M/3M.pdf}.

\bibitem{Mathoverflow}
The~Amplitwist (https://mathoverflow.net/users/183188/the amplitwist).
\newblock If a graph embedded on a surface is divided by a curve into a right
  and left that do not intersect can it be embedded on a surface of smaller
  genus?
\newblock MathOverflow.
\newblock URL:https://mathoverflow.net/q/458733 (version: 2023-11-19).
\newblock URL: \url{https://mathoverflow.net/q/458733}, \href
  {http://arxiv.org/abs/https://mathoverflow.net/q/458733}
  {\path{arXiv:https://mathoverflow.net/q/458733}}.

\bibitem{Karp1972}
Richard~M. Karp.
\newblock {\em Reducibility among Combinatorial Problems}, pages 85--103.
\newblock Springer US, Boston, MA, 1972.
\newblock \href {https://doi.org/10.1007/978-1-4684-2001-2_9}
  {\path{doi:10.1007/978-1-4684-2001-2_9}}.

\bibitem{KinseyTopology}
L.C. Kinsey.
\newblock {\em Topology of Surfaces}.
\newblock Undergraduate Texts in Mathematics. Springer New York, 1997.
\newblock URL: \url{https://books.google.gl/books?id=AKghdMm5W-YC}.

\bibitem{lau_ravi_singh_2011}
Lap~Chi Lau, R.~Ravi, and Mohit Singh.
\newblock {\em Iterative Methods in Combinatorial Optimization}.
\newblock Cambridge Texts in Applied Mathematics. Cambridge University Press,
  2011.
\newblock \href {https://doi.org/10.1017/CBO9780511977152}
  {\path{doi:10.1017/CBO9780511977152}}.

\bibitem{VLSIapp}
Charles~E. Leiserson and James~B. Saxe.
\newblock Retiming synchronous circuitry.
\newblock {\em Algorithmica}, 6(1):5--35, 1991.
\newblock \href {https://doi.org/10.1007/BF01759032}
  {\path{doi:10.1007/BF01759032}}.

\bibitem{ProgramVerifyApp}
Orna Lichtenstein and Amir Pnueli.
\newblock Checking that finite state concurrent programs satisfy their linear
  specification.
\newblock In {\em Proceedings of the 12th ACM SIGACT-SIGPLAN Symposium on
  Principles of Programming Languages}, POPL '85, page 97–107, New York, NY,
  USA, 1985. Association for Computing Machinery.
\newblock \href {https://doi.org/10.1145/318593.318622}
  {\path{doi:10.1145/318593.318622}}.

\bibitem{Tour2Apx}
Daniel Lokshtanov, Pranabendu Misra, Joydeep Mukherjee, Fahad Panolan,
  Geevarghese Philip, and Saket Saurabh.
\newblock 2-approximating feedback vertex set in tournaments.
\newblock {\em ACM Trans. Algorithms}, 17(2), apr 2021.
\newblock \href {https://doi.org/10.1145/3446969} {\path{doi:10.1145/3446969}}.

\bibitem{RydholmClassSur}
J.~Rydholm.
\newblock Classification of compact orientable surfaces using morse theory.
\newblock 2016.

\bibitem{schlomberg2023packing}
Niklas Schlomberg, Hanjo Thiele, and Jens Vygen.
\newblock Packing cycles in planar and bounded-genus graphs, 2023.
\newblock \href {http://arxiv.org/abs/2207.00450} {\path{arXiv:2207.00450}}.

\bibitem{ZuylenBip}
Anke van Zuylen.
\newblock Linear programming based approximation algorithms for feedback set
  problems in bipartite tournaments.
\newblock In {\em TAMC}, 2009.

\bibitem{YoungsTheoremPaper}
J.~W.~T. Youngs.
\newblock Minimal imbeddings and the genus of a graph.
\newblock {\em J. Math. Mech.}, 12:303--315, 1963.

\end{thebibliography}

\end{document}